\documentclass[a4paper,11pt]{article}
\pdfoutput=1

\usepackage{jheppub}
\usepackage{natbib}

\usepackage[utf8]{inputenc}

\usepackage{todonotes}
\usepackage{amssymb,amsmath,amsbsy,amsthm,mathtools}
\usepackage{braket}
\usepackage{graphicx,color}
\usepackage{multirow}
\usepackage{hyperref}
\usepackage{subcaption}
\usepackage{adjustbox}
\usepackage{relsize} 
\usepackage{xcolor}

\usepackage{array}

\newtheorem{definition}{Definition}[section]

\newtheorem{lemma}{Lemma}[section]
\newtheorem{theorem}{Theorem}[section]
\newtheorem{corollary}{Corollary}[section]
\newtheorem{construction}{Construction}[section]

\newtheorem{rules}{Rule}

\newcommand{\mL}{\mathcal{L}}

\newcommand{\mR}{\mathcal{R}}

\newcommand{\mbN}{\mathbb{N}}

\newcommand{\mbZ}{\mathbb{Z}}

\DeclareMathOperator{\area}{area}

\newcommand{\ms}{\mathsmaller}

\usepackage{graphicx}

\title{A framework for generalizing toric inequalities for holographic entanglement entropy}

\author[a,b]{Ning Bao}
\author[a]{Keiichiro Furuya}
\author[a,c]{Joydeep Naskar}

\affiliation[a]{Department of Physics, Northeastern University, Boston, MA, 02115, USA}
\affiliation[b]{Computational Science Initiative, Brookhaven National Laboratory, Upton, NY 11973 USA}
\affiliation[c]{The NSF AI Institute for Artificial Intelligence and Fundamental Interactions, Cambridge, MA, U.S.A.}

\emailAdd{ningbao75@gmail.com}
\emailAdd{k.furuya@northeastern.edu}
\emailAdd{naskar.j@northeastern.edu}

\abstract{

We conjecture a multi-parameter generalization of the toric inequalities of \cite{Czech:2023xed}. We then extend their proof methods for the generalized toric inequalities in two ways. The first extension constructs the graph corresponding to the toric inequalities and the generalized toric conjectures by tiling the Euclidean space. An entanglement wedge nesting relation then determines the geometric structure of the tiles. In the second extension, we exploit the cyclic nature of the inequalities and conjectures to construct cycle graphs. Then, the graph can be obtained using graph Cartesian products of cycle graphs. In addition, we define a set of knots on the graph by following \cite{Czech:2023xed}. These graphs with knots then imply the validity of their associated inequality. We study the case where the graph can be decomposed into disjoint unions of torii. Under the specific case, we explore and prove the conjectures for some ranges of parameters. We also discuss ways to explore the conjectured inequalities whose corresponding geometries are $d$-dimensional torii $(d>2)$.

}

\makeatletter
\gdef\@fpheader{}
\makeatother
\begin{document}

\maketitle

\section{Introduction}
The $AdS/CFT$ correspondence\cite{Maldacena1997} gives a dictionary between a theory of quantum gravity in the bulk in an asymptotically $AdS$-spacetime and a conformal field theory living on its boundary. However, not all quantum states on the boundary have a semi-classical dual geometry in the bulk. One of the foundational questions to address is then, which boundary states are dual to semi-classical bulk geometries?
Holographic entropy inequalities (HEIs) provide a partial answer to this question by imposing non-trivial constraints on the quantum states that are dual to semi-classical bulk theories.
For example, one of the simplest holographic entropy inequalities (which is not obeyed by all quantum states, but only holographic ones) is an inequality involving three regions, known as the monogamy of mutual information (MMI)\cite{MMI}:
\begin{equation}\label{eq:MMI}
    S(ab)+S(ac)+S(bc)\geq S(a)+S(b)+S(c)+S(abc).
\end{equation}
Such inequalities follow from the Ryu-Takayanagi(RT) formula\footnote{It is believed that all HEIs derived by the RT formula also hold for their covariant generalization to HRT formula\cite{mukund2007_HRT}. For example, see \cite{Czech:2019lps, white2024_hrt-HEIs} for the case of $2+1$ dimensions.}\cite{Ryu2006} that gives a bulk geometric interpretation of the boundary entanglement entropy or \textit{holographic entanglement entropy}(HEE). The entropy $S(\omega)$ of a boundary region $\omega$ is given by 
\begin{equation}\label{eq:RTformula}
    S(\omega)=\frac{\area \mathcal{X}_\omega}{4G_N},
\end{equation}
where $\mathcal{X}_\omega$ is the Ryu-Takayanagi surface for $\omega$ and $G_N$ is Newton's constant. 

The collection of all such maximally tight constraint HEIs form a rational polyhedral cone known as the holographic entropy cone(HEC), which was first studied for five regions in \cite{Bao:2015bfa}. Each of these HEIs are facets of the HEC and the full characterization of the HEC for five regions was completed in \cite{Sergio2019:5regions}.

Recently, a systematic search for candidate HEIs exploiting their structure in the $I$-basis\cite{He:2019repackaged} written as tripartite mutual information and tripartite conditional mutual information was pursued in \cite{Hernandez-Cuenca:2023iqh} leading to the discovery of 1877 previously unknown 6-party HEIs that are also facets of the HEC for six regions. A complete characterization of the HEC for six(and higher)-regions remains an open problem. A parallel development was the discovery of two infinite families\footnote{A single instance of which was first published in \cite{Czech:22seven-party}.} of HEIs\cite{Czech:2023xed} motivated from the holographic cone of average entropies (HCAE)\cite{Czech:2021average_cone,shec2021}. These inequalities are also facets of their respective HECs\cite{Czech:204facet_proof}. The toric inequalities are one such family (which we review in Section \ref{sec:review-toric}).

The standard method of proving a candidate HEI is the proof by contraction method\footnote{Some other methods for proving HEIs include the use of bit-threads\cite{Freedman_2016,headrick2022covariant, cui2019} and a more recent geometric proof for the MMI inequality \cite{Bousso_2024}.}, where one constructs a contraction map between binary hypercubes satisfying the occurrence data as initial constraints and the contraction condition. Traditionally, the computationally expensive greedy algorithm was used to generate the contraction maps. Recently, a more computationally efficient algorithm to construct contraction maps was given in \cite{Bao:2024contraction_map}. For special families, one may analytically design a deterministic strategy to construct the contraction map, but the applicability of these techniques is largely restricted to those special families. We will discuss such constructive proof methods based on graphical representations of HEIs and use them to prove some exemplary HEIs from our class of conjectures.

In this paper, we study the generalization of toric inequalities and report the existence of a family of HEIs uplifted from a subclass of toric inequalities, along with their detailed proof. The organization of this paper is as follows: In section \ref{sec:definitions}, we set up our notation for the paper. In section \ref{sec:review-toric}, we review the toric inequalities\cite{Czech:2023xed} and describe their graphical representations. In section \ref{sec:generalized-toric} we introduce a generalization of toric inequalities that conjectures new candidate inequalities and discuss the proof methods. In section \ref{sec:examples}, we prove some examples of true inequalities from those generalized toric conjectures and characterize them. Lastly, we discuss possible extensions to our work in section \ref{sec:discussions}.

\section{Definitions and Notations}\label{sec:definitions}

%\begin{center}
\begin{table}[t]
    \centering
    \begin{tabular}{|c||c|c|}
    \hline
     & \cite{Czech:2023xed} & This paper\\
    \hline
    \hline
       Boundary disjoint subregions  & $A_i,B_j$ & $a_i,b_j$ \\
    \hline
       Types of regions   & $A$-, $B$-type region& $(+)$-, $(-)$-region \\
    \hline
       Parameter of $(+)$-region  & $m$ & $\alpha$, $\alpha_s$ \\
    \hline
       Parameter of $(-)$-region  & $n$ & $\beta$, $\beta_t$ \\
    \hline
       Bitstrings  & $x,y$ & $x,y$ or $X,Y$ \\
    \hline
       Contraction maps  & $f$ & $f$ or $F$ \\
    \hline
       Graph in the proof by a contraction map  & N/A & $T_\mR$ \\
    \hline
       Vertices, edges, faces of $T_\mR$& $X,Y$ for all of them & $R_v$, $E_w$, $L_u$\\
    \hline
       HEE of a set $A=\{a_i\}_{i=1}^{\alpha}$& $S_{a_1^{(\alpha)}}$  & $S_A$\\
    \hline
    \end{tabular}
    \caption{Notation transitions}
    \label{tab:notations}
\end{table}
%\end{center}

We describe the notations used in the paper. Our notations are inspired by \cite{Czech:2023xed} with some changes. We summarize the notation transitions in table \ref{tab:notations} and graph theoretic notations in table \ref{tab:graph-notations}.
%For $N+1$ disjoint regions, let $A=\{a_i\}_{i=1}^{N+1}$ be a set of single disjoint regions including a purifier. %We define $\mP(N+1)$ as a power set of $\{a_i\}_{i=1}^{N+1}$.%

This paper deals with several subsets of disjoint subregions. For example, for $N+1$ disjoint \emph{monochromatic}\footnote{Consider a boundary region divided into $N+1$ subregions and each subregion is identified with a unique character. A \emph{monochromatic} subregion is labelled by a single character, whereas a \emph{polychromatic} subregion is, in general, labelled by multiple characters. %For example, if the boundary is divided into subregions $\{O, A, B, C\}$, each element of this set is a monochromatic index, whereas combinations like $AB, ABC$ are referred as polychromatic indices.
}subregions, we consider them as a union of two subsets of disjoint monochromatic subregions denoted as
\begin{equation}
    A=\{a_i\}_{i=1}^{\alpha}, \;B=\{b_j\}_{j=1}^{\beta}
\end{equation}
where $\alpha+\beta=N+1$, and $\alpha$ and $\beta$ are odd numbers. Here, $A$ and $B$ correspond to two disjoint sets of monochromatic subregions respectively. More generally, we can have $n_\alpha$ and $n_\beta$ number of disjoint sets of monochromatic subregions, i.e.,
\begin{equation}
    A_s = \{a_{(i_s,s)}\}_{i_s=1}^{\alpha_s}, \;B_t = \{b_{(j_t,t)}\}_{j_t=1}^{\beta_t}
\end{equation}
%where the indices run from $i_s=1,\cdots, \alpha_s$ for $s=1,\cdots, n_\alpha$ and $i_t=1,\cdots, \beta_t$ for $t=1,\cdots,n_\beta$.
where the indices run from $s=1,\cdots, n_\alpha$ and $t=1,\cdots,n_\beta$. $n_\alpha$ and $n_\beta$ are numbers of disjoint subsets of regions, respectively. For each $s$ and $t$, we have $i_s=1,\cdots, \alpha_s$ and $j_t=1,\cdots, \beta_t$ denoting the monochromatic regions of the $s$-th and $t$-th subsets respectively. We denote $a_{(i_s,s)}$ and $b_{(j_t,t)}$ as $a_{i_s}$ and $b_{j_t}$ whenever there is no ambiguity.
Here, the total number of regions are $\sum_{s=1}^{n_\alpha} \alpha_s + \sum_{t=1}^{n_\beta} \beta_t =N+1$. For example, when $n_\alpha=2$ and $n_\beta=1$, we have
\begin{equation}
\begin{split}
    A_1=\{a_{(1,1)},a_{(2,1)},&\cdots,a_{(\alpha_1,1)}\}, \;A_2=\{a_{(1,2)},a_{(2,2)},\cdots,a_{(\alpha_2,2)}\},\\
    &B_1=\{b_{(1,1)},b_{(2,1)},\cdots,b_{(\beta_1,1)}\}.\\
\end{split}
\end{equation}

In addition, we require the indices to satisfy the $\text{mod }\alpha_s$ and $\text{mod }\beta_t$ condition, i.e., 
\begin{equation}\label{eq:mod}
    i_s = i_s+\alpha_s  \text{ mod }\alpha_s, \;j_t = j_t+\beta_t \text{ mod }\beta_t, \; \forall i_s,j_t,s,t.
\end{equation}
We equivalently write the above as
\begin{equation}
    a_{(i_s,s)} \equiv a_{(i_s+\alpha_s,s)},\; b_{(j_t,t)} \equiv b_{(j_t+\beta_t,t)},\;\forall i_s,j_t,s,t.
\end{equation}
In general, we write a set of arbitrary \emph{polychromatic} boundary subregions\footnote{$\omega_k$ do not have to be disjoint to each other.} $\omega_k$ as $\Omega=\{\omega_k\}_{k=1}^{|\Omega|}$  where $|\Omega|$ is the cardinality of the set.

To simplify the notations, for $A$ and $B$, \cite{Czech:2023xed} introduced 
\begin{equation}\label{eq:a_i^k}
    a^{(k)}_i = a_i \cdots a_{i+k-1},\;b^{(k)}_j = b_j \cdots b_{j+k-1},
\end{equation}
and
\begin{equation}\label{eq:a_i^+}
    a^{\pm}_i := a_i^{(\frac{\alpha\pm 1}{2})},\;b^{\pm}_j := b_j^{(\frac{\beta\pm 1}{2})}.
\end{equation}
%We define $(+)$-regions by $A^+:=\{a^+_i\}_{i=1}^{\alpha}$ and $B^+:=\{b^+_j\}_{j=1}^{\beta}$, and $(-)$-regions by $A^-:=\{a^-_i\}_{i=1}^{\alpha}$ and $B^-:=\{b^-_j\}_{j=1}^{\beta}$. In the same way, $(+)$-regions and $(-)$-regions for $A_s$ and $B_t$ are denoted as $A^{\pm}_s:= \{a^{\pm}_{(i_s,s)}\}$ and $B^{\pm}_t:=\{b^\pm_{(j_t,t)}\}$.%
We define $(+)$-type regions by
\begin{equation}\label{eq:A+B+}
    A^+:=\{a^+_i\}_{i=1}^{\alpha}, \quad B^+:=\{b^+_j\}_{j=1}^{\beta},
\end{equation} 
and $(-)$-type regions by 
\begin{equation}
    A^-:=\{a^-_i\}_{i=1}^{\alpha}, \quad B^-:=\{b^-_j\}_{j=1}^{\beta}.
\end{equation}
More generally, $(+)$-type and $(-)$-type regions for $A_s$ and $B_t$ are denoted as $A^{\pm}_s:= \{a^{\pm}_{i_s}\}_{i_s=1}^{\alpha_s}$ and $B^{\pm}_t:=\{b^\pm_{j_t}\}_{j_t=1}^{\beta_t}$ respectively.

For any boundary subregion $\omega$, we may compute the entanglement entropy $S_\omega$. For example, $A=\{a_i\}_{i=1}^\alpha$, we denote $S_{a_{i}}$ to be the holographic entanglement entropy of a single disjoint subregion $a_{i}$. The entropy of a composite subregion $\{a_{i},a_{i+1}, a_{i+2}\}$ is written as $S_{a_{i}a_{i+1}a_{i+2}}$. In particular, we write the entropy of the set $A$ of regions as
\begin{equation}\label{eq:A}
    S_A := S_{a_1^{(\alpha)}}.
\end{equation}

We write a $N$-party entropy inequality for $N+1$ disjoint subregions (including the purifier)\footnote{``$N$-party entropy inequality'' implies that no term in the inequality contains a purifier explicitly.} as
\begin{equation} \label{eq:genentineq}
    \sum_{u=1}^l c_u S_{L_u} \geq \sum_{v=1}^r d_v S_{R_v}
\end{equation}
where $c_u, d_v >0$ are positive coefficients.  $L_u, R_v$ are the corresponding composite subregions of $u$-th term on the LHS and $v$-th term on the RHS, respectively. $l$ and $r$ are the total number of terms on the LHS and RHS. 

%We denote a set of the composite subregions of $u$-th terms $L_u$ on the LHS of the inequality as
We denote the set of all terms (to be precise, the subregions associated with them) $L_u$  on the LHS of the inequality as
\begin{equation}
    \mL:=\{L_u\}_{u=1}^{l}.
\end{equation}
Similarly, the set of associated subregions for all terms $R_v$ on the RHS of the inequality is denoted as
\begin{equation}
   \mR:=\{R_v\}_{v=1}^{r}. 
\end{equation}

%We denote $\mL \subset \mP(N+1)$ the subset of regions appearing on the left-hand side(LHS) of the inequality. Similarly, we denote $\mR$ for that of the RHS.

\begin{table}[t]
    \centering
    \begin{tabular}{| m{8cm} || m{3cm} | m{2.5cm} |}
    \hline
       $\alpha_s$-,$\beta_t$-cycle graph for $A_s^\pm$ and $B_t^\pm$ & $C^\pm_{\alpha_s}$, $C^\pm_{\beta_t}$  & Definition \ref{def:cyclegraphs}, (\ref{def:cyclegraphs-st})\\
    \hline
       Graph Cartesian product between $G$ and $H$  & $G \Box H$ & Definition \ref{def:complementgraphproduct} \\
    \hline
       Left and right graph  & $G_\mL$, $G_\mR$ & (\ref{eq:leftrightgraph})\\
    \hline
       Left and right toroidal graph& $T_\mL$, $T_\mR$ & Definition \ref{lem:embed-torus}  \\
    \hline
       $n_\alpha$-, $n_\beta$-toroidal graph& $C^\pm_{\alpha_1} \Box \cdots \Box C^\pm_{\alpha_{n_\alpha}}$, $C^\pm_{\beta_1} \Box \cdots \Box C^\pm_{\beta_{n_\beta}}$ & (\ref{eq:higher-dim-toroidal-graph})\\
    \hline
       Subgraph of $n_\alpha$-, $n_\beta$-toroidal graph& $G^\pm_{\{\alpha\}},G^\pm_{\{\beta\}}$ & (\ref{eq:multicyclegraphsalpha}),(\ref{eq:multicyclegraphsbeta}) \\
    \hline
       Cycle graphs in the decomposition of $G^\pm_{\{\alpha\}},G^\pm_{\{\beta\}}$ & $C^\pm_{\kappa_{\{\alpha\}}}, C^\pm_{\kappa_{\{\beta\}}}$ & Lemma \ref{lem:decomposition-with-cyclegraphs}\\
    \hline
      Toroidal graphs in the decomposition of $T_\mL$ and $T_\mR$ & $T_{\mL_\tau},T_{\mR_\tau}$,
      
      $\tau:= (\kappa_{\{\alpha\}},\kappa_{\{\beta\}})$ & (\ref{eq:leftrighttorus-decomposition})\\
    \hline
    \end{tabular}
    \caption{Graph notations used in the paper. The first and second columns present the graphs and their notations respectively. The third column points to their first introduction in the paper.}
    \label{tab:graph-notations}
\end{table}

\section{Review of the toric inequalities and their proof by a geometric contraction map}\label{sec:review-toric}

In this section, we review the infinite family of toric inequalities, characterized by two odd numbers $(\alpha,\beta)$ first found in \cite{Czech:2023xed}, followed by the analytical proof of toric inequalities being HEIs. Hence, we only deal with graphs embeddable into a two-dimensional Euclidean space.  

\subsection{Toric inequalities}

Consider two sets of disjoint regions $ A=\{a_i\}_{i=1}^{\alpha}$ and $B=\{b_j\}_{j=1}^{\beta}$, where one of them includes the purifier, i.e., the composite holographic quantum state over $(\alpha+\beta)$-regions is a pure state. 
The toric inequalities can be expressed as
\begin{equation} \label{eq:toric-definition}
    \sum_{i=1}^{\alpha}\sum_{j=1}^{\beta} S_{a_i^+b_j^-} \geq \sum_{i=1}^{\alpha}\sum_{j=1}^{\beta} S_{a_i^-b_j^-}+ S_{A},
\end{equation}
written in the notation defined in (\ref{eq:a_i^+}) and (\ref{eq:A}). Inequality (\ref{eq:toric-definition}) has a dihedral symmetry $D_{\alpha} \times D_{\beta}$ over the regions $A$ and $B$ respectively.
Replacing the terms explicitly containing the purifier $O$ with their complements yields a $(\alpha+\beta-1)$-party HEI. It has been further proved in \cite{Czech:204facet_proof} that the toric inequalities are the facets of HEC. As noted in \cite{Czech:2023xed}, this family subsumes the family of dihedral inequalities found in \cite{Bao:2015bfa}
\begin{equation}
    \label{eq:dihedral-ineqs}
    \sum_{i=1}^{\alpha} S_{a_i^+} \geq \sum_{i=1}^{\alpha} S_{a_i^-}+ S_{A}.
\end{equation}
It is therefore a natural curiosity to uplift\footnote{There is more than one way to uplift the toric inequalities. In section \ref{sec:generalized-toric}, we will consider the most straightforward generalization. In section \ref{sec:balance-superbalance}, we give examples of an extended class of generalizations that obey balance and superbalance conditions.} the toric inequality (\ref{eq:toric-definition}) to more regions to conjecture new holographic inequalities and test their validity.

\begin{figure}
    \centering
    \includegraphics[width=0.8\linewidth]{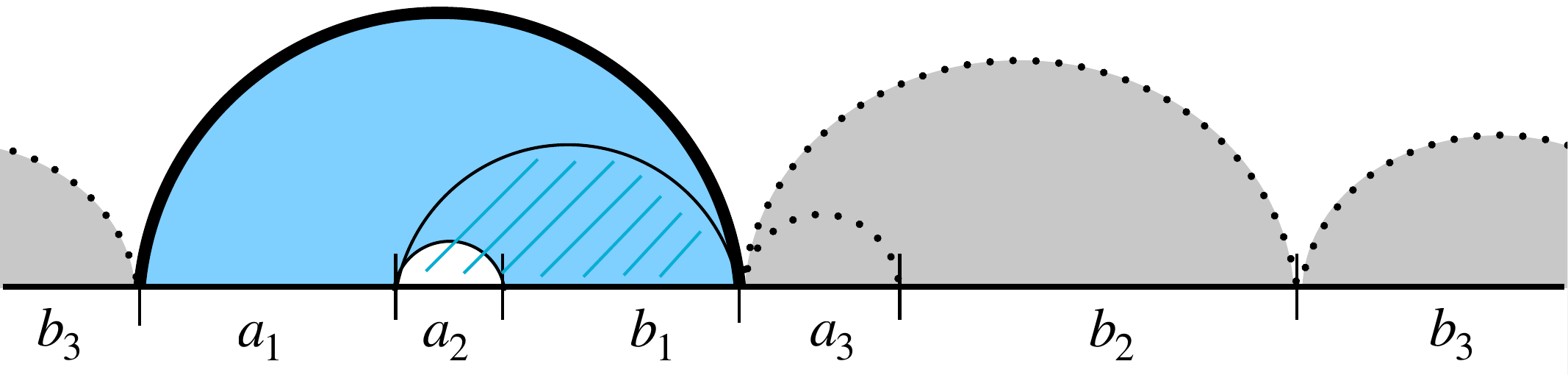}
    \caption{\small{The example of the inclusions and exclusions of entanglement wedges in a static time slice of AdS$_3$/CFT$_2$. The boundary region is partitioned into $a_1a_2a_3b_1b_2b_3$. The geometry with the partition is the example of $(3;3)$-inequalities. The boundary of entanglement wedge $EW(a_1a_2b_1)$ is the union of the thick solid black line and the boundary subregion $a_1a_2b_1$. The blue region is $EW(a_1b_1)$. $EW(a_2b_1)$ is covered with the blue lines. Any other entanglement wedges outside of $EW(a_1a_2b_1)$ are highlighted with gray, whose RT surfaces are the dotted lines.}}
     \label{fig:33etwnplanartorus}
\end{figure}

%In this subsection, we provide two ways of constructing the 

\subsection{Graphical representation of $(\alpha;\beta)$-inequalities}
The toric inequalities can be represented graphically for finite $\alpha$ and $\beta$\footnote{One may study the continuum limit setting $(\alpha,\beta)\to (\infty,\infty)$\cite{Czech:2023xed}.}. We provide two ways of constructing the graph. The first approach is similar to the one used in \cite{Czech:2023xed}, in which we construct atomic cells tiling the two-dimensional Euclidean space. In the second approach, we construct a toroidal graph\footnote{A graph that can be embedded on a torus.} exploiting the cyclic nature of the inequalities. 

The graph is constructed to satisfy an entanglement wedge nesting(EWN) relation. %, which we define in \ref{def:ewn}. 
The entanglement wedge $EW(\omega)$ of a boundary subregion $\omega$ is the bulk subregion enclosed by $\omega$ on the boundary and its RT surface in the bulk homologous to $\omega$. For two boundary subregions $\omega$ and $\omega'$ such that $\omega\subseteq \omega'$, EWN states that $EW(\omega)\subseteq EW(\omega')$\cite{akers2016geometricconstraintssubregionduality,Akers2020ewn}. We write 
\begin{equation}
    \omega \underset{\ms{EWN}}{\subseteq} \omega'
\end{equation}
for any pair of boundary subregions satisfying the EWN relation.

%Consider a set $\Omega=\{\omega_k\}_{k=1}^{|\Omega|}$ of arbitrary boundary subregions $\omega_k$. $\omega_k\in \Omega$ is said to be included in $\omega_l \in \Omega$ in the sense of EWN when $EW(\omega_k)$ is inside of $EW(\omega_l)$. We denote the relation as
%\begin{equation}
 %   \omega_k \underset{\ms{EWN}}{\subset} \omega_l
%\end{equation}
%for $k\neq l$.
For a set $\Omega=\{\omega_k\}_{k=1}^{|\Omega|}$ of arbitrary boundary subregions, we define inclusion subset $Inc(\omega_k)\subseteq \Omega$ and exclusion subset $Exc(\omega_k)\subseteq \Omega$ of a given arbitrary subregion $\omega_k$ using the EWN relation.
\begin{definition}
    Consider a set $\Omega=\{\omega_k\}_{k=1}^{|\Omega|}$ of arbitrary boundary subregions. For $\omega_k\in \Omega$, 
    \begin{equation}
    \begin{split}
        Inc(\omega_k) &:= \{ \omega_l \in \Omega|\; \omega_l \underset{\ms{EWN}}{\subseteq} \omega_k\}\\%\text{ for } \forall \omega_l \in \Omega  \},  \\
        Exc(\omega_k) &:= \{ \omega_l \in \Omega|\; \omega_l \underset{\ms{EWN}}{\subseteq} \bar{\omega}_k \}\\%\text{ for } \forall \omega_l \in \Omega  \},
    \end{split}
    \end{equation}
    where $\bar{\omega}_k$ denotes the complement of the boundary subregion $\omega_k$. $Inc(\omega_k)$ is the subset of boundary subregions whose entanglement wedges are contained in $EW(\omega_k)$. $Exc(\omega_k)$ is the subset of boundary subregions whose entanglement wedges are contained in $EW(\bar{\omega}_k)$. For later purposes, we define
    \begin{equation}
        \overline{Inc}(\omega_k) := \{ \omega_l \in \Omega|\; \omega_l \underset{\ms{EWN}}{\supseteq} \omega_k\} %\text{ for } \forall \omega_l \in \Omega  \}. 
    \end{equation}
    It is the subset of boundary subregions whose entanglement wedges contain $EW(\omega_k)$. 
\end{definition}

\begin{figure}
    \centering
    \includegraphics[width=0.8\linewidth]{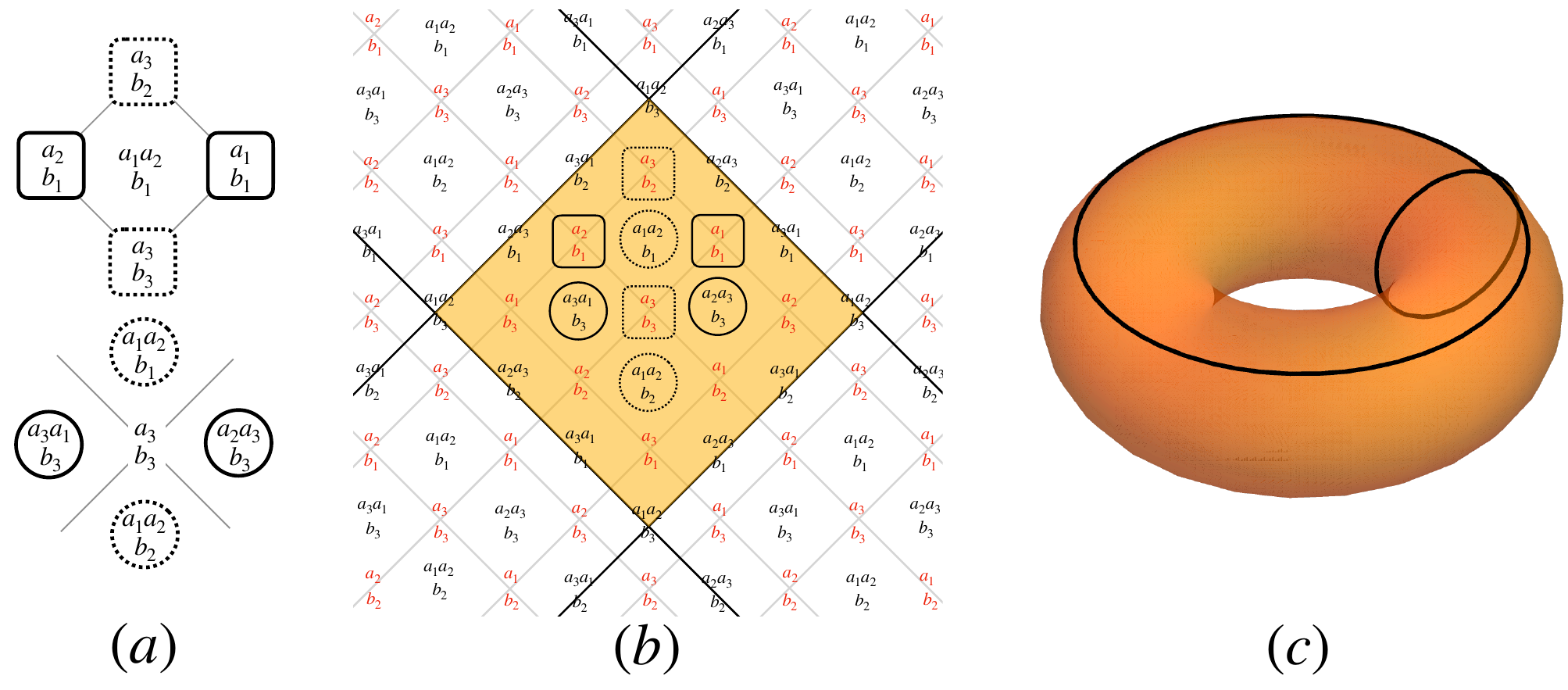}
    \caption{\small{(a) A choice of atomic cell. The solid rounded squares denote the elements of $Inc(a_1a_2b_1)=\{a_1b_1,a_2b_1\}$, and the dotted rounded squares denote the elements of $Exc(a_1a_2b_1)=\{a_3b_2,a_3b_3\}$. The solid circles are the elements of $\overline{Inc}(a_3b_3)=\{a_2a_3b_3,a_3a_1b_3\}$, and the dotted circles are the elements of $ Exc(a_3b_3) =\{a_1a_2b_1,a_1a_2b_2\}$. %The subscripts of B-regions change in the vertical direction. Those of $A$-regions change in the horizontal direction. 
    (b)Graphical representation of $(3;3)$-inequalities. The yellow ochre region is the fundamental domain. Identifying the parallel black lines gives a torus. (c) It is the torus in which the graph $(b)$ can be embedded. The black lines correspond to those in (b) running perpendicular to each other.}}
     \label{fig:33atomicnonplanartorus}
\end{figure}

To find the atomic cells, we study these subsets of the regions associated with the terms on the LHS and RHS of the $(\alpha;\beta)$-inequalities. For $(\alpha;\beta)$-inequalities in (\ref{eq:toric-definition}), consider $\mL =\{L_u\}_{u=1}^{\alpha \beta} $ and $\mR= \{R_v\}_{v=1}^{\alpha \beta+1}$ where $l=\alpha \beta$, $r=\alpha \beta+1$ and $\alpha+\beta =N+1$. The inclusion subset and the exclusion subset of $L_u \in \mL$ are 
\begin{equation}
\begin{split}
    Inc(L_u) &= \{ R_v \in \mR| R_v \underset{\ms{EWN}}{\subseteq} L_u \},\\%\text{ for } \forall R_v \in \mR  \} \subseteq \mR,\\
    Exc(L_u)&= \{ R_v \in \mR| R_v \underset{\ms{EWN}}{\subseteq} \bar{L}_u\}.%\text{ for } \forall R_v \in \mR  \}\subseteq \mR.
\end{split}
\end{equation}
For $R_v \in \mR$, we have%\footnote{}
\begin{equation}
\begin{split}
    \overline{Inc}(R_v) &= \{ L_u \in \mL| L_u \underset{\ms{EWN}}{\supseteq} R_v \} ,\\%\text{ for } \forall L_u \in \mL  \}\subseteq \mL , \\
    Exc(R_v) &= \{ L_u \in \mL| L_u \underset{\ms{EWN}}{\subseteq} \bar{R}_v\}.%\text{ for } \forall L_u \in \mL  \} \subseteq \mL.
\end{split}
\end{equation}

For example, in $(3;3)$-inequality\footnote{For the figure of $(5;3)$-inequality, see \cite{Czech:2023xed}. } as in figure \ref{fig:33etwnplanartorus} and \ref{fig:33atomicnonplanartorus}, we get 
\begin{equation}\label{eq:atomiccells}
\begin{split}
    Inc(a_i a_{i+1} b_j) &= \{a_i b_j, a_{i+1} b_j \},\\
    Exc(a_i a_{i+1} b_j) &= \{ a_{i+2}b_{j+1},a_{i+2}b_{j+2}\}\\
    \overline{Inc}(a_ib_j) &= \{a_{i-1}a_{i}b_j,a_ia_{i+1}b_j \},\\
    Exc(a_ib_j) &= \{a_{i+1}a_{i+2}b_{j+1},a_{i+1}a_{i+2}b_{j+1} \}\\
\end{split}
\end{equation}
for $\forall i,j=1,2,3$.
The cardinalities of the sets are given %\footnote{The general case is presented in lemma \ref{lem:cardinalities-general} in section \ref{sec:generalized-toric}.} 
by
\begin{equation}
     |Inc(L_u)| = 2-\delta_{\alpha,1}, |Exc(L_u)| = 2-\delta_{\beta,1},\; \forall L_u \in \mL
\end{equation}
and 
\begin{equation}
     |\overline{Inc}(R_v)| = 2-\delta_{\alpha,1}, |Exc(R_v)| = 2-\delta_{\beta,1},\; \forall R_v \in \mR
\end{equation}
where $\delta_{\alpha,1}$ and $\delta_{\beta,1}$ are the Kronecker deltas.

To construct an atomic cell, we choose to place $L_u$ on the face of a rhombus. We diagonally put the elements of $Inc(L_u)$ and $Exc(L_u)$ on its \emph{opposite} vertices. One can tile the two-dimensional space with these rhombi. The assignment of $\overline{Inc}(R_v)$ and $Exc(R_v)$ naturally follows from the assignment of $Inc(L_u)$ and $Exc(L_u)$ by the EWN relation, i.e., $R_v$ is placed on a vertex of the rhombus and the elements of $\overline{Inc}(R_v)$ and $ Exc(R_v)$ are on the faces of the rhombi adjacent to the vertex $R_v$, see figure \ref{fig:33atomicnonplanartorus} for instance.

The choice of geometry of an atomic cell is based on the fact that each LHS term $L_u$ is related to four terms on the RHS by EWN and vice-versa, i.e., $|Inc(L_u)|+|Exc(L_u)|=4$ and $|\overline{Inc}(R_v)|+| Exc(R_v)|=4$ (except the last term on the RHS, which is not a part of the tiling), for $\alpha>1$ and $\beta>1$\footnote{When either $\alpha=1$ or $\beta=1$, the single element appears repeatedly on an atomic cell. For example, $Exc(L_u)$ of $(3,1)$-inequality, or monogamy of mutual information, has a single element. We diagonally assign the single element repeatedly.}. %The diagonal assignments are essential to respect the dihedral symmetry $D_\alpha \times D_\beta$ %, especially cyclic symmetries,
%of the $(\alpha;\beta)$-inequalities in the geometry. We explicitly observe this by constructing a toroidal graph from cycle graphs.

\begin{figure}
    \centering
    \includegraphics[width=0.6\linewidth]{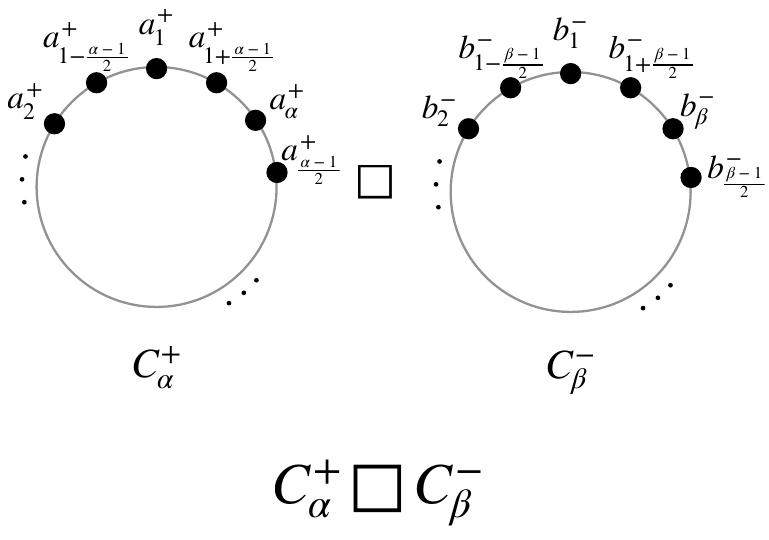}
    \caption{\small{Graph Cartesian product between $\alpha$-cycle graph $C^+_{\alpha}$ and $\beta$-cycle graph $C^-_{\beta}$. The black circle dots and the gray solid lines represent the vertices and the edges of the cycle graphs, respectively. }}
    \label{fig:cycleproducts1}
\end{figure}

\begin{figure}
    \centering
    \includegraphics[width=0.8\linewidth]{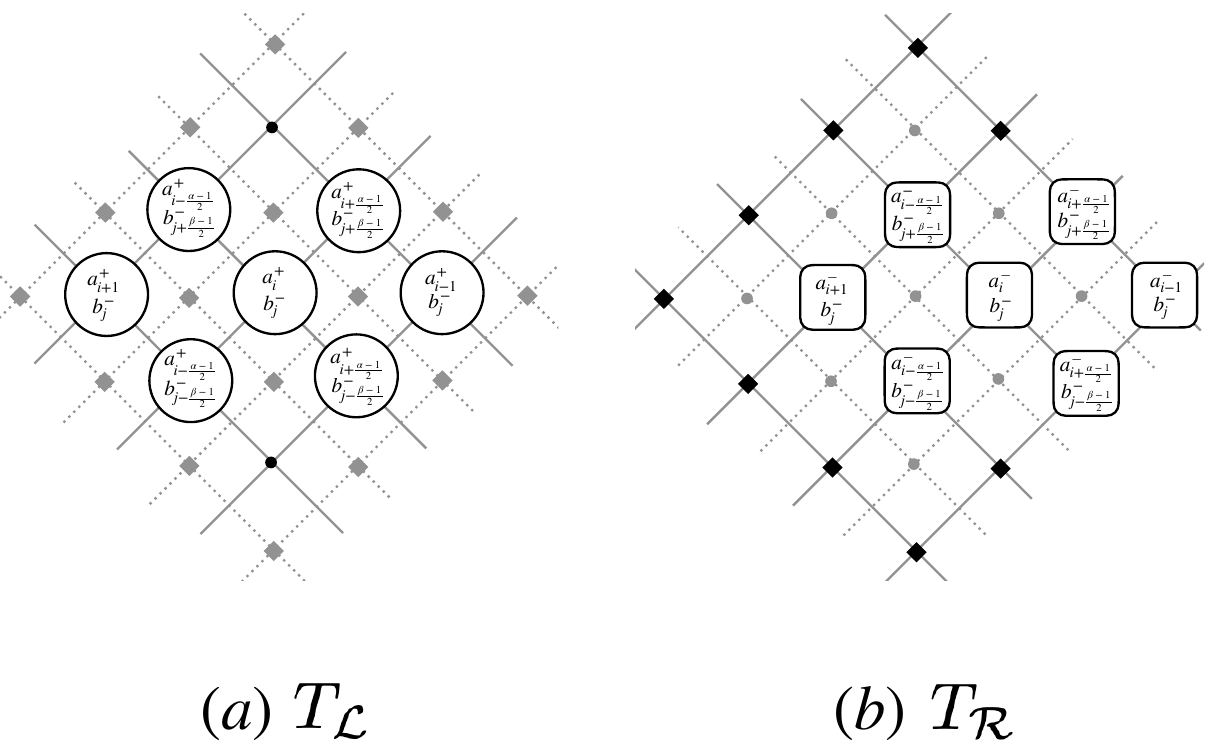}
    \caption{\small{(a) The left toroidal graph $T_\mL=(V_\mL,E_\mL,F_\mL)$. The black circle dots are the vertices $V_\mL$. Some vertices are labeled by the subregions. The gray solid lines are the edges $E_\mL$. The faces $F_\mL$ are enclosed by four dots and four edges. The dual vertices and the dual edges of $T_\mL$ are denoted by the gray square dots and the gray dotted lines. (b) The right toroidal graph $T_\mR = (V_\mR,E_\mR,F_\mR)$ dual to $T_\mL$. The black square dots are the vertices $V_\mR$. Some vertices are labeled by the subregions. The gray solid lines are the edges $E_\mR$. The faces $F_\mR$ are enclosed by four dots and four edges. The dual vertices and the dual edges of $T_\mR$ are denoted by the gray circle dots and the gray dotted lines.}}
    \label{fig:cycleproducts2}
\end{figure}

Below, we summarize the first construction.
\begin{construction}[Geometrization of $(\alpha;\beta)$-inequalities: Tiling method, figure \ref{fig:33atomicnonplanartorus}] \label{thm:tile-torus}\

    Consider $(\alpha;\beta)$-inequality.
    \begin{enumerate}
        \item For $\mL$ and $\mR$, determine the four sets $Inc(L_u),Exc(L_u),\overline{Inc}(R_v), Exc(R_v)$. 
        \item Construct the atomic cells in two-dimensional space by placing the pair of the elements of $Inc(L_u)$ and that of $ Exc(L_u)$ diagonally around $L_u$ for each $L_u\in\mL$.% such that they respect the cycles $[A]^-_\alpha$, $[B]^-_\beta$.
        \item Tile two-dimensional space by combining the atomic cells by identifying the vertices with the same label.
    \end{enumerate}

\end{construction}

For the second construction, we first define cycle graphs from the sets $A^\pm$ and $B^\pm$, which encode the cyclic symmetries in $(\alpha;\beta)$-inequality to the graph. 
\begin{definition}[$\alpha$- and $\beta$-cycle graphs, figure \ref{fig:cycleproducts1}]\label{def:cyclegraphs}

    For $A^\pm= \{a^\pm_i\}_{i=1}^{\alpha}$ and $B^\pm= \{b^\pm_j\}_{j=1}^{\beta}$, we define odd cycle graphs, $C^\pm_\alpha$, $C^\pm_\beta$%\footnote{Hamiltonian graphs are the graphs whose all the vertices in a cycle are arrived exactly once.}
    , of unit distance as
    \begin{equation}
        C^\pm_\alpha := \{V^\pm_\alpha,E^\pm_\alpha\}, \; C^\pm_\beta := \{V^\pm_\beta,E^\pm_\beta\}.
    \end{equation}
    $V^\pm_\alpha := A^\pm$ and $V^\pm_\beta := B^\pm$ are the sets of vertices labeled by the boundary subregions. $E^\pm_\alpha$ and $E^\pm_\beta$ are the sets of edges \footnote{We denote edges between the vertices $a_i$ and $a_j$ by the brackets $\braket{a_i,a_j}$.} defined as
    \begin{equation}
    \begin{split}
        E^\pm_\alpha&:=\{\braket{a^\pm_i,a^\pm_{i'}}| i'-i= \frac{\alpha-1}{2}\text{ or }i'-i=-  \frac{\alpha-1}{2} \}\\
        E^\pm_\beta&:=\{\braket{b^\pm_j,b^\pm_{j'}}| j'-j= \frac{\beta-1}{2}\text{ or }j'-j=-\frac{\beta-1}{2} \}\\
    \end{split}
    \end{equation}
    
\end{definition}
The number of vertices and edges of $C^\pm_\alpha$ and $C^\pm_\beta$, and their graph lengths $|C^\pm_\alpha|$ and $|C^\pm_\beta|$ are $|C^\pm_\alpha|=|V^\pm_\alpha| = |E^\pm_\alpha|=\alpha$ and $|C^\pm_\beta|=|V^\pm_\beta| = |E^\pm_\beta|=\beta$, respectively. To construct the graphs $G_\mL$ and $G_\mR$ that have the information on the LHS and the RHS of the $(\alpha;\beta)$-inequalities, we define graph Cartesian product below.
\begin{definition}[Graph Cartesian product\cite{gross1998graph,Bondy-Murty}]\label{def:complementgraphproduct}

    %For a simple graph\footnote{A simple graph has neither self-loops nor multiedges\cite{alma9937731250001401}. A multi-edge is a collection of edges ending on the same pair of vertices. We only have simple graphs in the paper.} $G=(V_G,E_G)$ and a set $K$ of all pairs of vertices in $V_G$, graph complement of $G$ is 
    %\begin{equation}
    %    \overline{G}: = \{V_G,K_G\setminus E_G \}.
    %\end{equation}
    %where $\overline{E}$ is a set of edges that are not the edges of $G$, i.e.
    %\begin{equation}
    %    \overline{E} = \{(g_1,g_2)|(g_1,g_2)\notin E,\; g_1,g_2 \in V\}.
    %\end{equation}

    For two graphs $G=(V_G,E_G)$ and $H=(V_H,E_H)$, we define the graph Cartesian product, i.e.,
     \begin{equation}
        G \Box H := (V_{GH},E_{GH}).
    \end{equation}
    The set of vertices $V_{GH} := V_{G}\times V_H \ni (g,h)$ is a Cartesian product between the set of vertices $V_{G}$ and $V_{H}$. $E_{GH}:= (V_G \times E_H )\cup (E_G \times V_H)$ is the set of edges.
    %:= (V_G \times E_H )\cup (E_G \times V_H)
    
    %Then, we denote $\Box$ the graph complement of the graph product between $G$ and $H$, i.e., 
    %\begin{equation}
    %    G \Box H : = \overline{G\Box H}  .
    %\end{equation}

\end{definition}
We then construct a \textit{left graph} $G_\mL$ and a \textit{right graph} $G_\mR$ up to graph isomorphisms\footnote{Two simple graphs $G$ and $H$ are graph isomorphic if the adjacency and non-adjacency vertices in $G$ are bijectively mapped to the adjacency and non-adjacency vertices in $H$, respectively.} by
\begin{equation}\label{eq:leftrightgraph}
\begin{split}
    G_\mL &\simeq C^+_\alpha \Box C^-_\beta = (V_\mL,E_\mL) \\
    G_\mR &\simeq C^-_\alpha \Box C^-_\beta =(V_\mR,E_\mR).
\end{split}
\end{equation}
where $|V_\mL|=|V_\mR|=\alpha\beta$ and $|E_\mL|=|E_\mR|=2\alpha\beta$, see figure \ref{fig:cycleproducts1}.
Note that for every $(i,j)$, where $i=1,\cdots, \alpha; j=1,\cdots, \beta$, there exists a corresponding $(a^+_i,b^-_j)\in V_\mL$ in $\mL$. Hence, we identify $V_\mL$ as $\mL$, and will use them interchangeably i.e., $V_\mL\equiv\mL$. Similarly, $V_\mR$ is identified as $\mR$, i.e., $V_\mR\equiv\mR$.

%It is useful to notice that $G_\mR$ is a dual graph of $G_\mL$ since it implies that $G_\mR$ can be constructed from $G_\mL$.
%\begin{lemma}

  %  $G_\mR$ is a dual graph of $G_\mL$ denoted by
   % \begin{equation}
    %    G_\mR = \hat{G}_\mL.
    %\end{equation}
    %Moreover, the pairs of non-adjacent vertices $A^+_iB^-_j \in \mL$ and $A^+_{i+1}B^-_j\in \mL$ for $i=1,\cdots, \alpha$ and $j=1,\cdots,\beta$ faithfully determine all the vertices in $\mR$. That is, 
    %\begin{equation}
    %    \mR =\{ R_v|R_v = A^+_iB^-_j \cap A^+_{i+1}B^-_j,\; \forall i=1,\cdots, \alpha, \forall j=1,\cdots,\beta  \}
    %\end{equation}
    %where $\cap$ is the intersection between $A^+_iB^-_j$ and $A^+_{i+1}B^-_j$ in the sense of boundary regions.

    %$:=  R_v\underset{\ms{EWN}}{\subset} A^+_iB^-_j\text{ and } R_v \underset{\ms{EWN}}{\subset} A^+_{i+1}B^-_j$

    %$R_v \underset{\ms{EWN}}{\subset} L_u\text{ and }  R_v \underset{\ms{EWN}}{\subset} L_{u'}$.
    %\begin{equation}
    %    R_v = L_u \underset{\ms{EWN}}{\cap} L_{u'}
    %\end{equation}
    %\begin{equation}
    %    \mR=\{R_v | R_v \underset{\ms{EWN}}{\subset} L_u\text{ and }  R_v \underset{\ms{EWN}}{\subset} L_{u'}\}
    %\end{equation}
    
%\end{lemma}
%\begin{proof}

 %   Without considering the labels of boundary subregions, both $G_\mL$ and $G_\mR$ are 

%\end{proof}

%physical motivation of the product, EWN
%Physically, the complement is needed to respect the EWN.

$G_\mL$ and $G_\mR$ can be embedded on a torus because one can draw each of them on a torus without any intersecting edges and $4$-cycles form the faces on the torus\cite{west_introduction_2000}. Now we can define \emph{toroidal graphs}.
\begin{definition}[Left toroidal graph and right toroidal graph, figure \ref{fig:cycleproducts2}]\label{lem:embed-torus}

    We denote $T_\mL$ and $T_\mR$ the left toroidal graph and the right toroidal graph obtained by embedding $G_\mL$ and $G_\mR$ on a $2$-torus, i.e.,
    \begin{equation}
        G_\mL \to T_\mL,\; G_\mR \to T_\mR ,
    \end{equation}
    and defined as 
    \begin{equation}
        T_\mL := (V_\mL,E_\mL,F_\mL), \;T_\mR := (V_\mR,E_\mR,F_\mR)
    \end{equation}
    where the sets, $F_\mL, F_\mR$, of labeled faces are faithfully determined by their set of vertices, $V_\mL=\mL, V_\mR=\mR$, based on the EWN relation, i.e.,
    \begin{equation}
    \begin{split}
        F_\mL = \underset{i=1,\cdots,\alpha,j=1,\cdots,\beta}{\bigcup} Inc(a^+_{i}b_j^-) \cap  Inc(a^+_{i+1}b_j^-)
        %F_\mL=\{a^-_{i+1}b^-_j \in \mR|& a^-_{i+1}b^-_j \underset{\ms{EWN}}{\subseteq} a^+_i b^-_j \in \mL \text{ and } a^-_{i+1}b^-_j \underset{\ms{EWN}}{\subseteq} a^+_{i+1}b^-_j\in \mL\;\\
        %&\text{ for }\forall i=1,\cdots, \alpha\text{ and } \forall j=1,\cdots,\beta  \}\\
    \end{split}
    \end{equation}
    \begin{equation}
    \begin{split}
        F_\mR = \underset{i=1,\cdots,\alpha,j=1,\cdots,\beta}{\bigcup} \overline{Inc}(a^-_{i}b_j^-) \cap  \overline{Inc}(a^-_{i+1}b_j^-)
        %F_\mR =\{a^+_{i}b^-_j \in \mL |& a^+_{i}b^-_j \underset{\ms{EWN}}{\supseteq} a^-_{i}b^-_j\in \mR\text{ and }a^+_{i}b^-_j \underset{\ms{EWN}}{\supseteq} a^-_{i+1}b^-_j\in \mR \; \\
        %&\text{ for }\forall i=1,\cdots, \alpha\text{ and
        %} \forall j=1,\cdots,\beta  \}\\
    \end{split}
    \end{equation}

    %where $\underset{\ms{EWN}}{\cap}$ is the intersection 
    %$F_\mL:= \{R_v\in \mR \}$ and $F_\mR:= \{L_u\in \mL \}$ are the set of faces.

    %Moreover, the pairs of non-adjacent vertices $A^+_iB^-_j \in \mL$ and $A^+_{i+1}B^-_j\in \mL$ for $i=1,\cdots, \alpha$ and $j=1,\cdots,\beta$ faithfully determine all the faces vertices in $\mR$. That is, 
    %\begin{equation}
    %    \mR =\{ R_v|R_v = A^+_iB^-_j \cap A^+_{i+1}B^-_j,\; \forall i=1,\cdots, \alpha, \forall j=1,\cdots,\beta  \}
    %\end{equation}
    %where $\cap$ is the intersection between $A^+_iB^-_j$ and $A^+_{i+1}B^-_j$ in the sense of boundary regions.

    %duality preserves
    %$ T_\mR \simeq \hat{T_\mL}$
    
\end{definition}
%\begin{proof}

 %   From \cite{}, they are embeddable.

  %  There are $\alpha\beta$ distinct pairs. Thus, the EWN relation determines all the $\alpha\beta$ distinct faces.

%\end{proof}

%\begin{corollary}

 %   $T_\mR$ is a dual toroidal graph of $T_\mL$. 
    
%\end{corollary}

One can check that $T_\mL$ and $T_\mR$ are topologically $2$-torus by calculating its graph genus $g$ from $|V|-|E|+|F|=2-2g$\cite{west_introduction_2000}, i.e., 
\begin{equation}
    g_\mL= g_\mR = \frac{|V_\mL|-|E_\mL|+|F_\mL|+2}{2} = 1
\end{equation}
where we used that $|V_\mL|=|V_\mR|=\alpha\beta$,  $|E_\mL|=|E_\mR|=2\alpha\beta$, and $|F_\mL|=|F_\mR|=\alpha\beta$.

It is useful to notice that $T_\mR$ is a dual graph of $T_\mL$.
\begin{lemma}

    $T_\mR$ is a dual graph of $T_\mL$ denoted by
    \begin{equation}
        T_\mR = \hat{T}_\mL.
    \end{equation}

\end{lemma}
The above lemma holds because there exists a one-to-one map $V_\mL\to F_\mR$, $E_\mL\to E_\mR$, and $F_\mL\to V_\mR$, see figure \ref{fig:cycleproducts2}. %Note that this duality is stronger than the notion of a dual graph in the sense that it preserves the labels of vertices and faces. 
Note that, for any vertex $L_u \in V_\mL$ of $T_\mL$, its dual vertex is a face of $T_\mR$ labeled with an element of $Inc(L_u)\cup Exc(L_u)$. Similarly, for any vertex $R_v \in V_\mR$ of $T_\mR$, its dual vertex is a face of $T_\mL$ labeled with an element of $\overline{Inc}(R_v)\cup Exc(R_v)$. 

%\begin{equation}
    %    \mR=\{R_v | R_v \underset{\ms{EWN}}{\subset} L_u\text{ and }  R_v \underset{\ms{EWN}}{\subset} L_{u'}\}
%\end{equation}

%Furthermore, the right toroidal graph $T_\mR$ is a dual graph of the left toroidal dual $T_\mL$, which we write by
%\begin{equation}
 %   T_\mR = \hat{T}_\mL
%\end{equation}
%because there exists one-to-one map $V_\mL\to F_\mR$, $E_\mL\to E_\mR$, and $F_\mL\to V_\mR$, see figure \ref{fig}. Note that the duality is stronger than a mere graph dual in the sense that the labels of the vertices in $G_\mR$ are determined by the EWN relation. In particular, for a vertex $L_u \in V^\mL_{\alpha\beta}$, its dual vertices are $Inc(L_u)\cup Exc(L_u)$. 

%\begin{lemma}[Right toroidal graph $T_\mR$ is a dual left graph $T_\mL$]

 %   Right toroidal graph $T_\mR$ is a dual graph to $T_\mL$.
    
%\end{lemma}

One should notice that the last term on the RHS of the toric inequality, $S_{A}$, does not belong to any of the above sets. This term is not geometrized explicitly but has a key role when constructing a geometric contraction map. We summarize the second construction below.

\begin{construction}[Geometrization of $(\alpha;\beta)$-inequalities: Graph theoretical method, figure \ref{fig:cycleproducts2}] \label{thm:graph-torus}\

    Consider $(\alpha;\beta)$-inequality.
    \begin{enumerate}
        \item Construct $\alpha$-, and $\beta$-cycle graphs, $C^+_\alpha$ and $C^-_\beta$.
        \item Construct $G_\mL = C^+_\alpha \Box C^-_\beta$. %$G_\mR=\hat{G}_\mL$ by taking graph dual of $G_\mL = G^+_\alpha \Box G^-_\beta$.
        %\item Construct $G_\mL = G^+_\alpha \Box G^-_\beta$ and $G_\mR = G^-_\alpha \Box G^-_\beta$.
        \item Obtain $T_\mL$ by embedding $G_\mL$ on a torus.
        \item $T_\mR = \hat{T}_\mL$ 
    \end{enumerate}

    Note that the graph explicitly used for the proof by a contraction map in \cite{Czech:2023xed} and this paper corresponds to $T_\mR$.

\end{construction}

%By placing these elements diagonally, the vertical and horizontal directions of the whole geometry respect the cyclic part of the symmetries of the inequalities. Encoding these two cycles in the geometrization is the key step to represent the toric inequality which has the dihedral symmetries $D_\alpha \times D_\beta$ as a discrete torus geometrically. 

%based on entanglement wedge nesting, resulting in a discretized torus, see figure \ref{fig:33etwnplanartorus}. Consider a set of subregions that are inside of entanglement wedge of the subregion

\subsection{Proof by a geometric contraction map}
In general, for a $N$-party entropy inequality involving $N+1$ disjoint regions (including the purifier), $A=\{a_i\}_{i=1}^{N+1}$,
\begin{equation} \label{eq:genentineq}
    \sum_{u=1}^l c_u S_{L_u} \geq \sum_{v=1}^r d_v S_{R_v},
\end{equation}
we define $l$- and $r$- dimensional bitstrings as $x \in \{0,1\}^l$ and $y\in \{0,1\}^r$. We define a special set of bitstrings, called \textit{occurrence bitstrings}, $x^{a_i} \in \{0,1\}^l$ for every single region $i\in \{1,\cdots,N+1\}$ and accordingly define $y^{a_i}$ in the RHS, as follows:
\begin{equation}
    (x^{a_i})_u =
\begin{cases}
1 & \text{if $a_i \subseteq L_u$ }\\
0 & \text{otherwise}\\
\end{cases}, \;
    (y^{a_j})_v =
\begin{cases}
1 & \text{if $a_j \subseteq R_v$}\\
0 & \text{otherwise}\\
\end{cases}.
\end{equation}
We denote the bitstrings that are not occurrence bitstrings without superscript, i.e., as $x$ and $y$ for bitstrings in LHS and RHS respectively. The norm of these bitstrings is defined by the weighted Hamming norms, i.e., %for instance, 
\begin{equation}
    \|x\|_c = \sum_{u=1}^{l} c_u |(x)_u|,\;   \|y\|_d = \sum_{v=1}^{r} d_v |(y)_v|.
\end{equation}
Generally, we use the following theorem to prove that entropy inequality (\ref{eq:genentineq}) is HEI.
\begin{theorem}[`Proof by contraction']\label{thm:proofbycontraction}\cite{Bao:2015bfa}
    Let $f:\{0,1\}^l \to \{0,1\}^r $ be a $(\|\cdot\|_c - \|\cdot\|_d)$-contraction, i.e.,
    \begin{equation}
        \|x -x' \|_c \geq \|f(x) -f(x') \|_d ,\; \forall x,x'\in \{0,1\}^l.
    \end{equation}
    If  $f(x^{a_i})  = y^{a_i}$ for $\forall i\in \{1,\cdots, N+1\}$, %\footnote{Because of these boundary conditions, occurrence bitstrings are the special bitstrings that set the initial data when proving the inequalities by a contraction map.},
    then (\ref{eq:genentineq}) is a HEI.
\end{theorem}

%Fortunately, we only need to check if the map $f$ is a contraction for all pairs $(x,x')$ with $\|x-x'\|_1 =1$, where the subscript $\|.\|_1$ denotes that all the Hamming weights are unity. This is summarized below.
%\begin{lemma}{\cite{Bao:2015bfa}}\label{lem:contraction}
%    For $x\in \{0,1\}^l$, define $f:\{0,1\}^l \to \{0,1\}^r $. If $f$ is a contraction for all the pairs $(x,x')$ such that $ \|x-x'\|_1 =1$, i.e.,
 %   \begin{equation}
 %       \|x-x'\|_1\geq \|f(x)-f(x')\|_1, \forall (x,x')\in  \{(x,x') \;|\; \|x-x'\|_1 =1, \; \forall x,x'\},
 %   \end{equation}
 %   then, $f$ is a contraction for all $x,x'$, i.e.
 %   \begin{equation}
  %      \|x-x'\|_1\geq \|f(x)-f(x')\|_1, \forall x,x'.
  %  \end{equation}
    
%\end{lemma}

%In the case of the toric inequalities, or $(\alpha;\beta)$-inequalities (\ref{eq:toric-definition}), we consider two sets $A=\{a_i\}_{i=1}^\alpha$ and $B=\{b_j\}_{j=1}^\beta$ of disjoint boundary subregions. We denote the occurrence bitstrings as either $x^{a_i}$ and $y^{a_i}$, and $x^{b_j}$ and $y^{b_j}$, respectively. $c_u$ and $d_v$ are $1$. Thus, the weighted Hamming distance becomes $\|x\|_c = \|x\|_1  = \sum_{u=1}^{l} |(x)_u|$ and $\|y\|_d = \|y\|_1 = \sum_{v=1}^{r} |(y)_v|$. In addition, $l=\alpha \beta$ and $r=l+1$.

$f(x^{a_i})  = y^{a_i} \text{ for } \forall i\in \{1,\cdots, N+1\}$ are the \textit{boundary conditions} through which the contraction map learns about the inequality. In the case of $(\alpha;\beta)$-inequalities (\ref{eq:toric-definition}), we have two sets $A=\{a_i\}_{i=1}^\alpha$ and $B=\{b_j\}_{j=1}^\beta$ of disjoint boundary subregions. The LHS and RHS occurrence bitstrings of $A$ are denoted as $x^{a_i}$ and $y^{a_i}$ respectively. Similarly, for $B$, we denote them by $x^{b_j}$ and $y^{b_j}$. For such inequalities, $l=\alpha \beta$ and $r=l+1$, in addition, $c_u$ and $d_v$ are all unity. Thus, the weighted Hamming distance becomes $\|x\|_c = \|x\|_1  = \sum_{u=1}^{l} |(x)_u|$ and $\|y\|_d = \|y\|_1 = \sum_{v=1}^{r} |(y)_v|$, where the subscript $\|.\|_1$ denotes that all Hamming weights are unity.

To prove that the toric inequalities (\ref{eq:toric-definition}) are HEIs, we use the following lemma (first introduced in \cite{Bao:2015bfa}),
%, we only need to check if the map $f$ is a contraction for all pairs $(x,x')$ with $\|x-x'\|_1 =1$. %, where the subscript $\|.\|_1$ denotes that all the Hamming weights are unity. 
\begin{lemma}{\cite{Bao:2015bfa}}\label{lem:contraction}
    For $x\in \{0,1\}^l$, define $f:\{0,1\}^l \to \{0,1\}^r $. If $f$ is a contraction for all the pairs $(x,x')$ such that $ \|x-x'\|_1 =1$, i.e.,
    \begin{equation}
        \|f(x)-f(x')\|_1\leq 1, \forall x,x'\in \{0,1\}^l \;s.t.\; \|x-x'\|_1 =1,
    \end{equation}
    then, $f$ is a contraction for all $x,x'$, i.e.
    \begin{equation}
        \|x-x'\|_1\geq \|f(x)-f(x')\|_1, \forall x,x'.
    \end{equation}
    
\end{lemma}

%To prove that the $(\alpha;\beta)$-inequalities are HEIs, 
Then, \cite{Czech:2023xed} proved the following theorem to prove the $(\alpha;\beta)$-inequalities (\ref{eq:toric-definition}) being HEIs. We review their proof in our language.
\begin{theorem}{\cite{Czech:2023xed}}\label{thm:toricproof}
    %For $x$, $x'$ such that $\|x-x'\|_1 =1$, a map $f:\{0,1\}^l \to \{0,1\}^r$ satisfies 
    Consider a contraction map $f: \{0,1\}^l \to \{0,1\}^r$ satisfying $f(x^{a_i})= y^{a_i}$ and $f(x^{b_j})= y^{b_j}$ for $\forall i\in \{1,\cdots, \alpha\}$ and $\forall j\in \{1,\cdots, \beta\}$ for an $(\alpha;\beta)$-inequality. If $\|x-x'\|_1 = 1 $ for $x,x'\in \{0,1\}^l$, then, $\|f(x)-f(x')\|_1 = 1 $.
    %, i.e.,
    %\begin{equation}
    %    \|x-x'\|_1 = 1 \to \|f(x)-f(x')\|_1 = 1, \; \forall x,x'.
    %\end{equation}
\end{theorem}

\begin{proof}

%corresondence b/w L_u and x_u, and R_v and y_v
$(\alpha;\beta)$-inequalities are represented by the toroidal graph $T_\mR$ on a $2$-torus described in the previous subsection. We geometrize the bitstrings $x$ and $y$ to construct a contraction map geometrically to prove the theorem. %In general, we color each face $L_u\in \mL=F_\mR$ and each vertex $R_v\in \mR=V_\mR$ of $T_\mR=(V_\mR,E_\mR,F_\mR)$ with a bit $\{0,1\}$. A candidate geometric contraction map acts on a set of bitstrings $x$ and determines a set of bitstrings $y$. %The proof is completed if it is a contraction map for all pairs of bitstrings with Hamming distance one and satisfying the boundary conditions.

\begin{figure}
    \centering
    \includegraphics[width=0.6\linewidth]{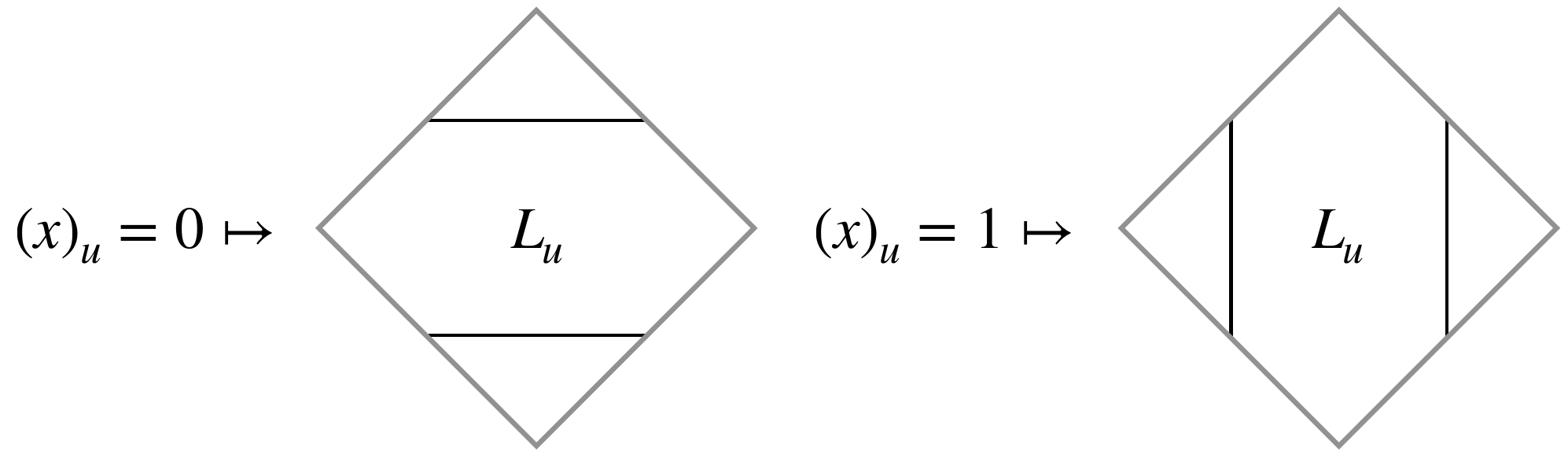}
    \caption{\small{A face $L_u$ is colored with $0$ or $1$ if the $u$-th element $(x)_u$ of bitstring $x$ is $0$ or $1$. When $(x)_u=0$, we add two parallel horizontal line segments connecting adjacent edges. When $(x)_u=1$, we add two parallel vertical line segments connecting adjacent edges.}}
    \label{fig:bitstringgeometry}
\end{figure}

Consider a bitstring $x =\{0,1\}^l$. %Since each $u$-th element $(x)_u$ corresponds to a face $L_u\in\mL$ of $T_\mR$, 
We color each face $L_u\in\mL$ of $T_\mR$ with a bit $\{0,1\}$. When a face is colored with $0$, we assign the two horizontal parallel line segments connecting the middle of the adjacent edges. When it is colored with $1$, then we assign the two vertical parallel line segments, see figure \ref{fig:bitstringgeometry}. Thus, for a given $x$, the connected line segments give a set $\Gamma(x)$ of knots\footnote{We call it 'knots' instead of 'loops' used in \cite{Czech:2023xed}.} on the torus. %We call the atomic cell with $x_u=0$ assigned $0$-\textit{face} and $x_u=1$ assigned $1$-\textit{face}.

These line segments form a set $\Gamma(x)$ of $(p,q)$-torus knots on the right toroidal graph $T_\mR$ where $p$ is a winding number around the longitudinal direction, and $q$ is a winding number around the meridian direction. In general, there are two types of knots: i) non-contractible $(p,q)$-knots $K_{nc}$ for $p\neq 0$ or $q\neq0$, and ii) contractible $(0,0)$-knots, $K_{c}$. These knots run over only faces and edges, not over vertices by the construction. Furthermore, there are no intersecting knots. 

%For a pair of binary vectors $x,x'$, $|x-x'|=1$ means the flipping 

The geometric assignment of $y\in\{0,1\}^r$ follows from the set of \textit{rules} given in \cite{Czech:2023xed} and summarized below. This also defines the geometric contraction map $f$. Let us call the first $l$ bits as \textit{geometric bits}, which correspond to the vertices of the $T_\mR$. We call the last single bit as a \textit{non-geometric bit} because the last bit does not explicitly appear in the geometry. Thus, the candidate geometric contraction map can be represented by
\begin{equation} \label{eq:contractionmap-toric}
    f(x) = \tilde{f}(x) \times \tilde{f}_{ng}
\end{equation}
where $\tilde{f}(x)\in\{0,1\}^l$ is the $l$-tuple of geometric bits and $\tilde{f}_{ng}\in\{0,1\}^{r-l=1}$ is the non-geometric bit corresponding to $S_A$ in (\ref{eq:toric-definition}). Below, we give the rules to read off the bitstrings from the color configuration of the faces $F_\mR$.

\begin{figure}
    \centering
    \includegraphics[width=0.5\linewidth]{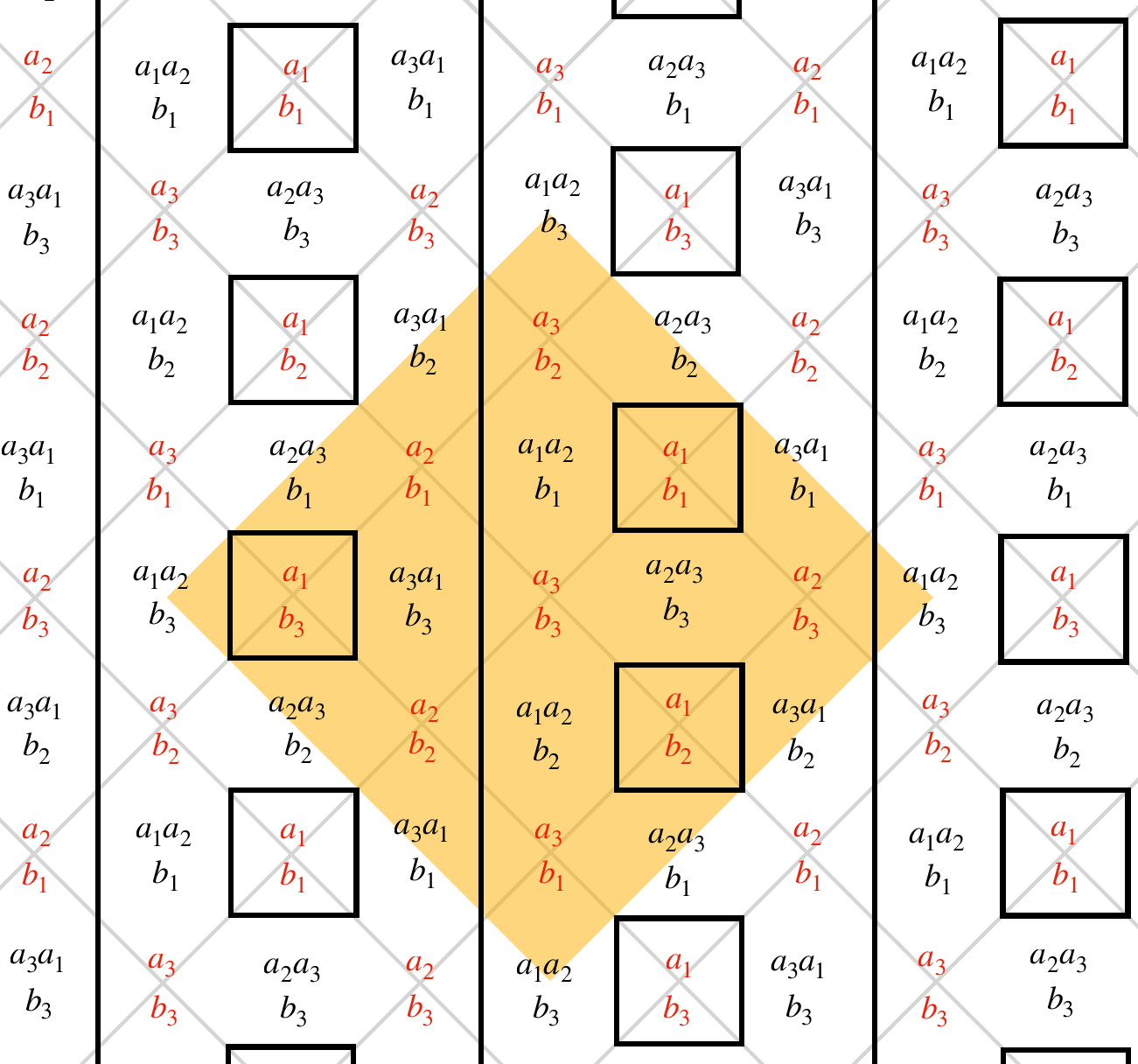}
    \caption{\small{Boundary condition, $f(x^{a_1})=y^{a_1}$, of $(3;3)$-inequality: The yellow ochre region is the fundamental domain of the torus. Black lines are the knots. $1$ is assigned to all the vertices enclosed by the contractible knots; otherwise, $0$.   }}
    \label{fig:boundarycondition}
\end{figure}

\begin{rules}\label{rul:toric}\

\noindent
\underline{1. Rules on geometric bits $\tilde{f}(x)$, vertices}

On the toroidal graph $T_\mR$ of $(\alpha;\beta)$-inequalities,
\begin{enumerate} \label{rule1}
    \item[a)] Assign $1$ to a vertex $R_v\in V_\mR= \mR$ enclosed by a contractible knot. 

    \item[b)] If a knot encloses multiple vertices, assign $1$ to the right-most and bottom-most vertex. Assign $0$ to all the other vertices inside the knot.
\end{enumerate}

\noindent
\underline{2. Rule on the non-geometric bit $\tilde{f}_{ng}$}

%The $r$-th element 
The non-geometric bit $\tilde{f}_{ng}$ is determined by
\begin{equation} \label{eq:toricrule2}
    \|x\|_1=\|f(x)\|_1\ \text{mod}\ 2.
\end{equation}
    %\item When $n>1$, there are more than $1$ in the non-geometric terms on the right-hand side. In this case, you have a choice. One makes a choice which is consistent with the boundary conditions.

    %how many freedom are there? how certain it is?  
\end{rules}
By Rule \ref{rul:toric}, the map $f$ satisfies the boundary conditions $f(x^{a_i})=y^{a_i}$ and $f(x^{b_j})=y^{b_j}$, see figure \ref{fig:boundarycondition}.

Now, we show that $\|x-x'\|_1=1$ implies $\|f(x)-f(x')\|_1=1$. The geometrical operation of $x\mapsto x'$ such that $\|x-x'\|_1=1$ corresponds to a bit-flip, i.e., $0\to 1$, on a single face. This results in either a knot sum of adjacent knots on the torus or splitting one knot into two knots, see figure \ref{fig:knotsums}. For two knots $K_1$ and $K_2$, the knot sum connects two knots, joining them by two bars and denoted by
\begin{equation}
    K_1 \# K_2.
\end{equation}
$\|f(x)-f(x')\|_1$ counts the change of number of contractible knots before and after the bit-flip. This implies that only one contractible knot can appear or disappear after the geometric operation. Since we only have two types of non-intersecting knots, contractible ones $K_c$ and non-contractible ones $K_{nc}$, there are only four types of knot sum as follows, see figure \ref{fig:knotsums}. For contractible knots $K_c$, $K'_c$, $K''_c$ and non-contractible knots $K_{nc}$, $K'_{nc}$, 
\begin{equation}
    K_c \#K'_c =K_c'',\;K_{nc} \#K_{c} =K'_{nc},\;K_{nc} \#K'_{nc} =K''_{c},\;K_{nc} \#K_{nc} =K'_{nc}.
\end{equation}
where the last sum represents the self-knot sum. The number of contractible knots changes by $1$,  except for the last sum. In these cases, the $r$-th non-geometric bit does not flip by the second rule (\ref{eq:toricrule2}). Hence, $\|x-x'\|_1=1$ implies $\|f(x)-f(x')\|_1=1$.

For the self-knot sum, the number of contractible knots does not change. By the second rule (\ref{eq:toricrule2}), the $r$-th non-geometric bit flips. Thus, $\|f(x)-f(x')\|_1=1$. Therefore, $\|x-x'\|_1=1$ implies $\|f(x)-f(x')\|_1=1$.

\end{proof}

Lemma \ref{lem:contraction} and Theorem \ref{thm:toricproof} prove the following corollary.
\begin{corollary}\cite{Czech:2023xed} \label{cor:toric-HEIs}
    Toric inequalities, or $(\alpha;\beta)$-inequalities are HEIs.
\end{corollary}

%\subsection{CMI edges}

\begin{figure}[t]
    \centering
    \includegraphics[width=0.8\linewidth]{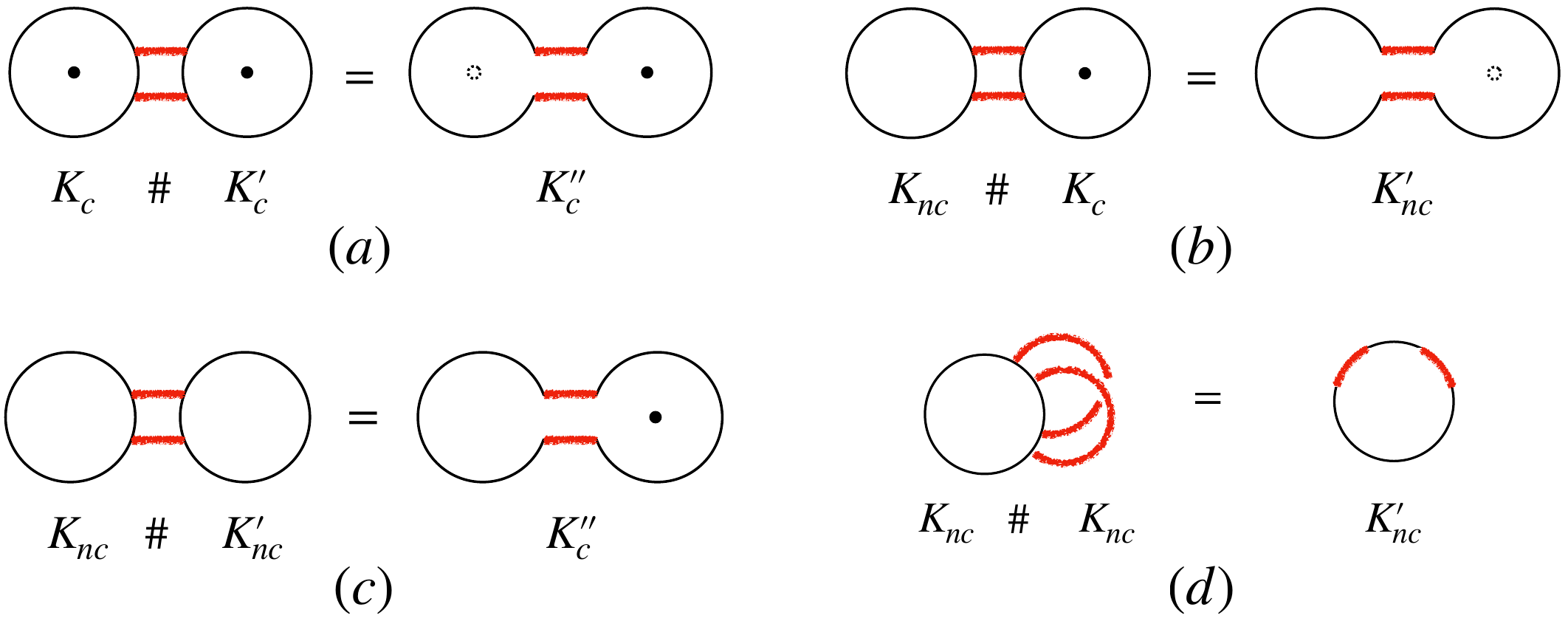}
    \caption{\small{All possible knot sums of non-intersecting knots. Black lines are the knots. They consist of the sequence of line segments defined by the geometrization of a LHS bitstring. The contractible knots have at least a single vertex colored with $1$. Thus, we represent the contractible knots with a black dot at their center. The non-contractible knots do not have a black dot at their center. The vertex color changes from $0$ to $1$ or vice versa after the knot sum, represented by a dotted circle. The red line segments define the knot sum between two knots. The knot sum between (a)two contractible knots, (b)a non-contractible knot and a contractible knot, (c)two non-contractible knots, and (d)itself.}}
    \label{fig:knotsums}
\end{figure}

\section{Generalized Toric Conjectures }\label{sec:generalized-toric} %(modify global terms)

We define the generalized toric conjectures, or $(\alpha_1,\cdots,\alpha_{n_\alpha};\beta_1,\cdots,\beta_{n_\beta})$-conjectures, to be
\begin{equation}\label{eq:inequality-gen-toric-extension}
\begin{split}
    \sum_{i_1,\cdots, i_{n_\alpha}=1}^{\alpha_1,\cdots,\alpha_{n_\alpha}} \sum_{j_1,\cdots, j_{n_\beta}=1}^{\beta_1,\cdots,\beta_{n_\beta}} S_{a^+_{i_1}\cdots a^+_{i_{n_\alpha}} b^-_{j_1}\cdots b^-_{j_{n_\beta}}}  \geq \sum_{i_1,\cdots, i_a=1}^{\alpha_1,\cdots,\alpha_{n_\alpha}} &\sum_{j_1,\cdots, j_b=1}^{\beta_1,\cdots,\beta_{n_\beta}}  S_{a^-_{i_1}\cdots a^-_{i_{n_\alpha}} b^-_{j_1}\cdots b^-_{j_{n_\beta}}}\\
    & + \textit{``non-geometric terms''}
\end{split}
\end{equation}
%where the purifier $O$ is in one of the $A$-type regions.

%Note that the generalized toric conjectures for all ranges of the parameters are not valid
Note that the generalized toric conjectures for any arbitrary parameters do not necessarily give a valid HEI \footnote{Some conjectures failed to produce a valid contraction map. By the completeness argument\cite{Bao:2024contraction_map}, they are ruled out as invalid.}. The terms that fill \textit{``non-geometric terms''} in Section \ref{sec:balance-superbalance} are determined, but not exhaustively, by the balance and superbalance conditions \cite{Bao:2015bfa,HEarrangement2019Hubeyetal,HErepackaged2019He,superbalance2020He,Hernandez-Cuenca:2023iqh}.

In this section, we first extend the geometric contraction map of toric inequalities, or $(\alpha;\beta)$-inequalities, to the generalized toric conjectures. Then, we apply the extended geometric contraction map to a few examples of the generalized toric conjectures in section \ref{sec:examples}.% to exercise a few example searches on new holographic entropy inequalities.

%The geometrization of the inequalities is fairly straightforward. One should note that only the last terms are not geometric.

\subsection{Graphical representation of $(\alpha_1,\cdots,\alpha_{n_\alpha};\beta_1,\cdots,\beta_{n_\beta})$-conjectures}\label{sec:proof-methods}
%We first extend theorem \ref{thm:tile-torus} to the generalized toric conjectures, see theorem \ref{thm:tile-general}. Then, we extend theorem \ref{thm:graph-torus}, see theorem \ref{thm:graph-general}. Although both approaches construct the same graph in general, the latter simplifies the proof.  
We generalize the constructions \ref{thm:tile-torus} and \ref{thm:graph-torus} to constructions \ref{thm:tile-general} and \ref{thm:graph-general} respectively, as applicable to the generalized toric conjectures. Although both approaches construct the same graph, the latter simplifies the proof.

Before we proceed to the construction, we state the basic property regarding the sets $Inc,Exc,\overline{Inc}$ for $(\alpha_1,\cdots,\alpha_{n_\alpha};\beta_1,\cdots,\beta_{n_\beta})$-conjectures.

\begin{lemma}\label{lem:cardinalities-general}

    Consider $(\alpha_1,\cdots,\alpha_{n_\alpha};\beta_1,\cdots,\beta_{n_\beta})$-conjectures. The cardinalities of the sets $Inc(L_u)$, $Exc(L_u)$, $\overline{Inc}(R_v)$, $Exc(R_v)$ are given by
\begin{equation}\label{eq:cardinalitylhs}
     |Inc(L_u)|  = \sum_{s=1}^{n_\alpha}(2-\delta_{\alpha_s,1}), |Exc(L_u)|=\sum_{t=1}^{n_\beta}(2-\delta_{\beta_t,1}),\;  \forall L_u \in \mL
\end{equation}
and 
\begin{equation}\label{eq:cardinalityrhs}
     |\overline{Inc}(R_v)|=\sum_{s=1}^{n_\alpha} (2-\delta_{\alpha_s,1}), |Exc(R_v)|=\sum_{t=1}^{n_\beta}(2-\delta_{\beta_t,1}) ,\; \forall R_v \in \mR
\end{equation}
where $\delta_{\alpha_s,1}$ and $\delta_{\beta_t,1}$ are the Kronecker deltas.
\end{lemma}
%\begin{proof}

%Each $A^+$ and $A^-$ has $2-\delta$ from (\ref{eq:cardinalitylhs}) and (\ref{eq:cardinalityrhs}). We have sums.

%\end{proof}
    %Do I need the proof for this in the paper?? or make it concise?
    %, i.e., for $\alpha_{s} \geq 2$ $(s=1,\cdots,a)$ and $\beta_{t} \geq 2$ $(t=1,\cdots,b)$,

 %   rewrite
    
 %   We first prove that $|Inc|=2^a$. All the term has the same number of subregions, which is $\frac{\alpha+1}{2}$ in each term by the definition of the conjectures. For each $\alpha$, the difference between the number of subregions in $L_u$ and $R_v$ is $1=\frac{\alpha+1}{2}-\frac{\alpha-1}{2}$. With the cyclic symmetries, only two $R_v$ can be included in $L_u$ for any $\alpha$. For each $\beta$, the difference is $0=\frac{\beta-1}{2}-\frac{\beta-1}{2}$. Hence, only a single subset of regions is included in $L_u$ for any $\beta$. Thus, the number of $R_v$ one can find in $L_u$ is $2^a$. 

 %   Similarly, we prove $|Exc(L_u)|=2^b$. In the compliment of $L_u$, the difference is $0=\frac{\alpha-1}{2}-\frac{\alpha-1}{2}$ for $\alpha$. For $\beta$, the difference is $1=\frac{\beta-1}{2}-\frac{\beta-1}{2}$. Thus, the number of $R_v$ one can find in the complement of $L_u$ is $2^b$. We skip the proof of (\ref{eq:cardinalityrhs}) since the logic is essentially the same.
    
%\end{proof}

For $(\alpha_1,\cdots,\alpha_{n_\alpha};\beta_1,\cdots,\beta_{n_\beta})$-conjectures, we construct the atomic cells and the whole geometry in $(n_\alpha+n_\beta)$-dimensional space. %The coordinates of the space correspond to the set $A_1=\{A_{i_1}\}_{i_1=1}^{\alpha_1}, \cdots,A_a=\{A_{i_a}\}_{i_a=1}^{\alpha_a}, B_1=\{B_{j_1}\}_{j_1=1}^{\beta_1}\cdots B_b=\{B_{j_b}\}_{j_b=1}^{\beta_b}$. See figure for the examples, $(\alpha_1,\alpha_2;\beta)$-conjectures and $(\alpha_1,\alpha_2,\alpha_3;\beta)$-conjectures.
The geometrization of $ (\alpha_1, \cdots, \alpha_{n_\alpha}; \beta_1, \cdots, \beta_{n_\beta}) $-conjectures is similar to the one given in construction \ref{thm:tile-torus}. However, the construction of atomic cells needs extra care since $|Inc(L_u)|$ and $|Exc(L_u)|$ are generally more than $2$. We choose the geometry of atomic cells such that the atomic cells in a two-dimensional subspace of $(n_\alpha+n_\beta)$-dimensional space reduce to rhombi. Then, we place the pair of elements of $Inc(L_u)$ and that of $Exc(L_u)$ diagonally on the vertices of the rhombus labeled with $L_u$ in the two-dimensional subspace spaces. We summarize the tiling construction below.

\begin{figure}
    \centering
    \includegraphics[width=0.8\linewidth]{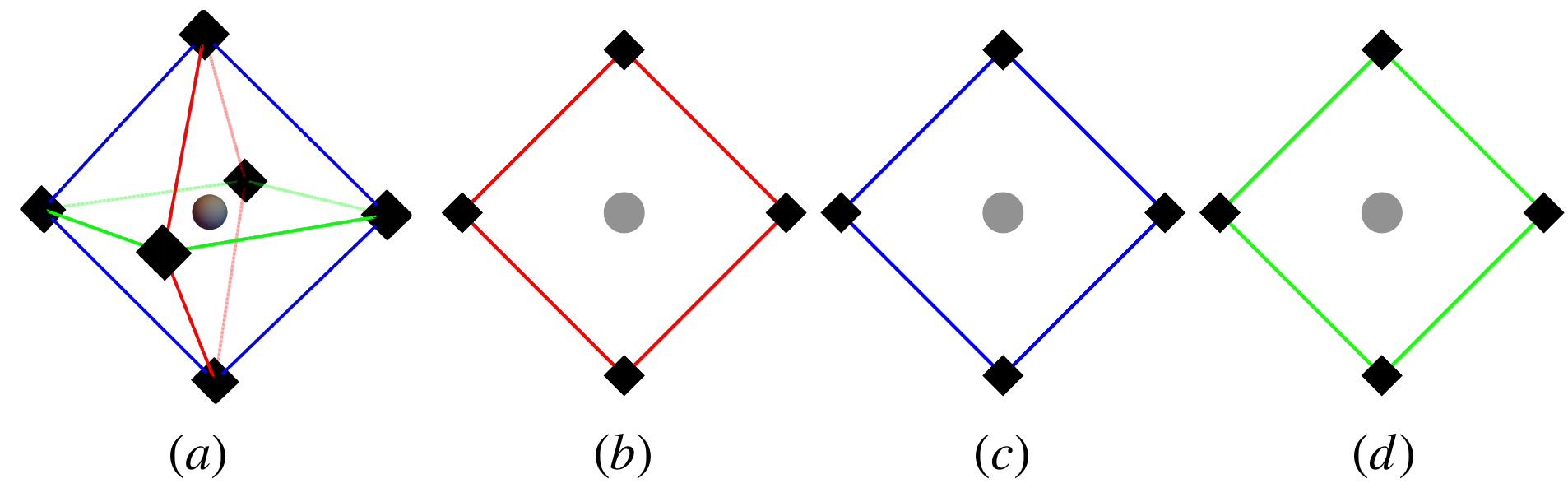}
    \caption{\small{For $n_\alpha=2$, $n_\beta=1$, (a) is an atomic cell in three-dimensional space such that it reduces to a rhombus in every two-dimensional space. The gray solid sphere is a volume of the octahedron to which $L_u$ are assigned. The black square dots are the vertices to which the elements of $Inc(L_u)$ and $Exc(L_u)$ are assigned. Here, $|Inc(L_u)|=4$ and $|Exc(L_u)|=2$ assuming that $\alpha_1,\alpha_2 \neq 1$ and $\beta_1 \neq 1$ for simplicity. The edges of each rhombus in each two-dimensional space of (a) are colored with (b) red, (c) blue, and (d) green. The choice of colors here is independent from the colors of the other figures and equations.}}
    \label{fig:21-tilingcell}
\end{figure}

\begin{construction}[Geometrizing the generalized toric conjectures: Tiling method, figure \ref{fig:21-tilingcell}] \label{thm:tile-general} \
     %\footnote{\label{limitation}Note that this construction is specific to a sub-class of conjectures as explained in footnote \ref{clarify-rhombi}.}
     Consider a $(\alpha_1,\cdots,\alpha_{n_\alpha};\beta_1,\cdots,\beta_{n_\beta})$-conjecture, and its $\mL$ and $\mR$.
    \begin{enumerate}
        \item For $\mL$ and $\mR$, determine the four sets $Inc(L_u),Exc(L_u),\overline{Inc}(R_v), Exc(R_v)$. 
        \item Consider $(n_\alpha+n_\beta)$-dimensional space whose coordinates correspond to the sets $A_1=\{a_{i_1}\}_{i_1=1}^{\alpha_1}, \cdots,A_{n_\alpha}=\{a_{i_{n_\alpha}}\}_{i_{n_\alpha}=1}^{\alpha_{n_\alpha}}, B_1=\{b_{j_1}\}_{j_1=1}^{\beta_1}, \cdots , B_{n_\beta}=\{b_{j_{n_\beta}}\}_{j_{n_\beta}=1}^{\beta_{n_\beta}}$.

        \item Choose the geometry of the atomic cells such that its tiling reduces to rhombi in every two-dimensional subspace. %\footnote{\label{clarify-rhombi}In the most general case, the atomic cell is a $({n_\alpha}+{n_\beta})$-dimensional polyhedron. Here, we are considering those cases, where this polyhedron can be decomposed into an union of rhombi across different planes.}
 
        \item Construct the atomic cells in the $(n_\alpha+n_\beta)$-dimensional space by placing the pair of the elements of $Inc(L_u)$ and that of $Exc(L_u)$ diagonally on the vertices of the rhombus labeled with $L_u$ in every two-dimensional subspace.%\footnote{Note that we are studying the special cases where the $n_{\alpha}+n_{\beta}$ dimensional space is tiled using 2-dimensional rhombi........} %complying with the $A^+$-cycles and $B^-$-cycles. The same applies to the elements of $Exc(L_u)$ around $L_u$. 

        \item Build the whole geometry by combining the atomic cells by identifying the vertices with the same label.
        
    \end{enumerate}
    
    %We denote the whole geometry of (\ref{eq:inequality-gen-toric-extension}) as $\Gamma(\alpha_1,\cdots,\alpha_a;\beta_1,\cdots,\beta_b)$.

     %geometric part of inequality corresponds to the conditional entropy $S(|) \geq \text{global terms}$. 
     
\end{construction}

%\begin{definition}[geometrizing bitstrings]
%    aaaaaa
%\end{definition}

%On $\Gamma(\alpha_1,\cdots,\alpha_a;\beta_1,\cdots,\beta_b)$, there are pairs of cycles that form torii. 

%1. existence of decompositions (check 5,3)

%2. define the geometric bitstrings on each torus

%3-1. discuss the interactions

%3-2. discuss there are no interactions among torii only when $(\alpha's, \beta's)$. 

%4. we work on non-interacting case $(\alpha,\cdots,\beta,\cdots)$ ones.

Now, we extend construction \ref{thm:graph-torus}. From definition \ref{def:cyclegraphs}, we define $\alpha_s$-, and $\beta_t$-cycle graphs, for $s=1,\cdots,{n_\alpha}$ and $t=1,\cdots,{n_\beta}$,
\begin{equation} \label{def:cyclegraphs-st}
     C^\pm_{\alpha_s} := \{V^\pm_{\alpha_s},E^\pm_{\alpha_s}\}, \; C^\pm_{\beta_t} := \{V^\pm_{\beta_t},E^\pm_{\beta_t}\}. 
\end{equation}
where $|C^+_{\alpha_s}| = |V^\pm_{\alpha_s}|=|E^\pm_{\alpha_s}| =\alpha_s$ and $|C^+_{\beta_t}| = |V^\pm_{\beta_t}|=|E^\pm_{\beta_t}| =\beta_t$.
%In this paper, instead of considering $D_{\alpha_1}\times \cdots D_{\alpha_a}\times D_{\alpha_1}\times \cdots D_{\alpha_a}$
With the graph Cartesian product in definition \ref{def:complementgraphproduct}, we get
\begin{equation} \label{eq:higher-dim-toroidal-graph}
\begin{split}
    &%(V^\pm_{\{\alpha\}},E^\pm_{\{\alpha\}}) \simeq 
    C^\pm_{\alpha_1} \Box \cdots \Box C^\pm_{\alpha_{n_\alpha}} , \; C^\pm_{\beta_1} \Box \cdots \Box C^\pm_{\beta_{n_\beta}}.\\
    %&(V^\pm_{\{\beta\}},E^\pm_{\{\beta\}}) \simeq 
    %C^\pm_{\beta_1} \Box \cdots \Box C^\pm_{\beta_{n_\beta}}
\end{split}
\end{equation}
Their set of vertices is defined by
\begin{equation}
    V^\pm_{\{\alpha\}}:= V^\pm_{\alpha_1}\times \cdots \times  V^\pm_{\alpha_{n_\alpha}}, \;V^\pm_{\{\beta\}}:= V^\pm_{\beta_1}\times \cdots \times  V^\pm_{\beta_{n_\beta}}.
\end{equation}
Their set of edges is defined by
\begin{equation} \label{eq:cycle-edges}
    \underset{s,s'=1,\cdots,{n_\alpha}}{\bigcup} \big( (V^\pm_{\alpha_s} \times E^\pm_{\alpha_{s'}}) \cup (E^\pm_{\alpha_s} \times V^\pm_{\alpha_{s'}})\big), \;\underset{t,t'=1,\cdots,{n_\beta}}{\bigcup} \big( (V^\pm_{\beta_t} \times E^\pm_{\beta_{t'}}) \cup (E^\pm_{\beta_t} \times V^\pm_{\beta_{t'}})\big)
\end{equation}
These graphs are the ${n_\alpha}$-dimensional toroidal graphs with $\prod_s \alpha_s$ vertices and $n_\alpha \prod_s \alpha_s$ edges and the ${n_\beta}$-dimensional toroidal graphs with $\prod_t \beta_t$ vertices and $n_\beta \prod_t \beta_t$ edges\cite{PARK202164,CHEN201433}, respectively. 

Then, we can get the $(n_\alpha + n_\beta)$-dimensional toroidal graph $\tilde{T}_\mR$ that matches with the one constructed by following construction \ref{thm:tile-general}, i.e.,
\begin{equation}
    \tilde{T}_\mR = \big( C^-_{\alpha_1} \Box \cdots \Box C^-_{\alpha_{n_\alpha}} \big) \Box \big( C^-_{\beta_1} \Box \cdots \Box C^-_{\beta_{n_\beta}}\big).
\end{equation}
In this paper, however, we restrict ourselves to a more straightforward case where the constructed graph decomposes into a disjoint union of $2$-torii\footnote{We comment on the general case in section \ref{sec:discussions}, and postpone it as a future exploration.}. This corresponds to picking a set of all the two-dimensional planes in $(n_\alpha+n_\beta)$-dimensional space, which do not cross each other. These two-dimensional planes are tiled only by, for instance, a rhombus (b) with the red edges in figure \ref{fig:21-tilingcell}.

\begin{figure}[t]
    \centering
    \includegraphics[width=15.7cm]{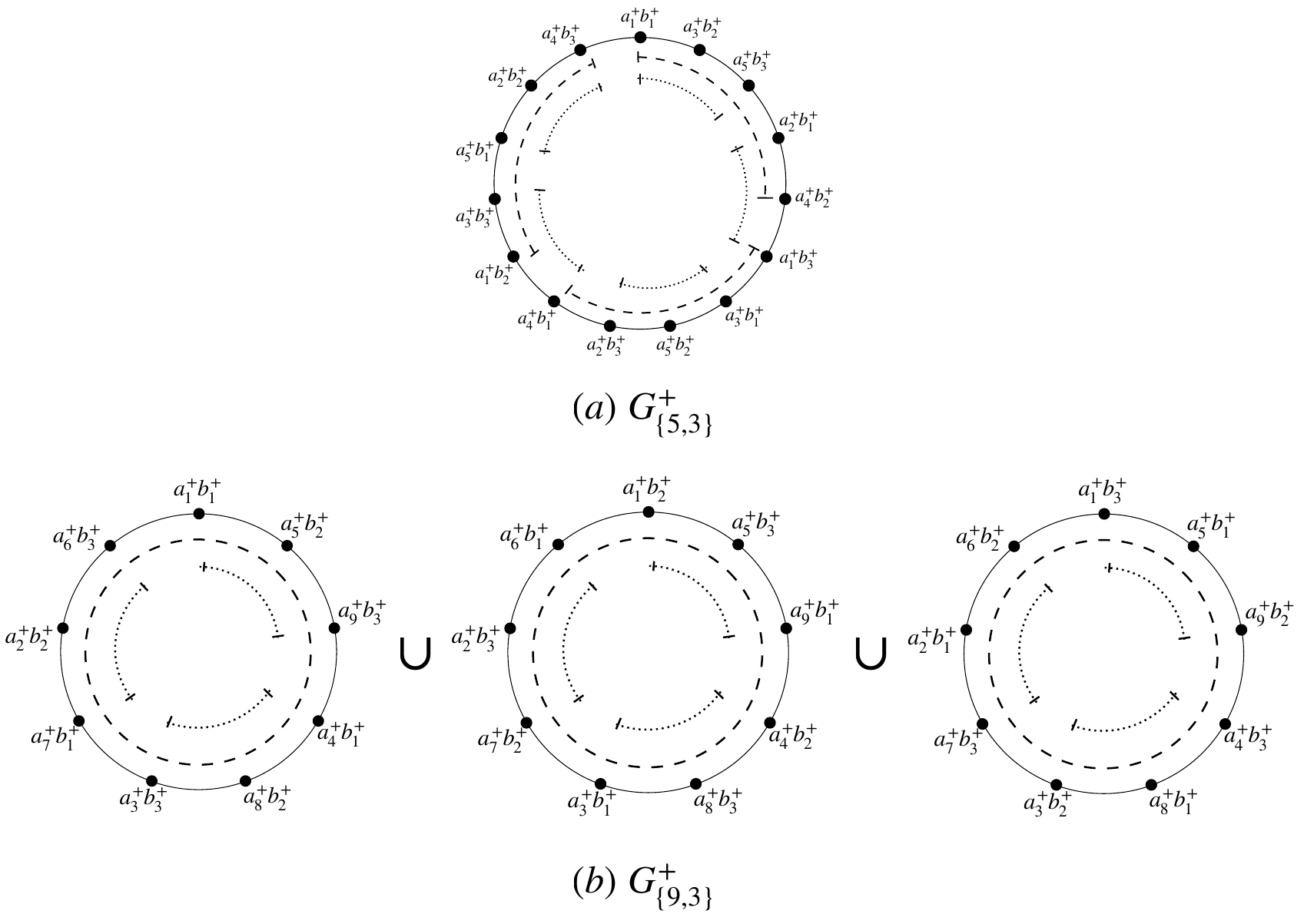}
    \caption{\small{Two examples of decomposition of a cycle graph constructed by the graph Cartesian product of two cycle graphs: (a) $G^+_{\{5,3\}}$ constructed by graph Cartesian product between $C_5^+$ of the set of regions $A^+=\{a^+_1,a_2^+,a^+_3,a^+_4,a^+_5\}$ and $C_3^+$ of the set of regions $B=\{b^+_1,b^+_2,b^+_3\}$ has $15/lcm(5,3)=1$ cycle graph with its graph length $|G^+_{\{5,3\}}|=lcm(5,3)=15$. The dashed line denotes a single period of $C_5^+$ and the dotted line represents a single period of $C_3^+$. (b) $G^+_{\{9,3\}}$ constructed by graph Cartesian product between $C_9^+$ of the set of regions $A^+=\{a^+_1,a^+_2,a^+_3,a^+_4,a^+_5,a^+_6,a^+_7,a^+_8,a^+_9\}$ and $C_3^+$ of the set of regions $B=\{b^+_1,b^+_2,b^+_3\}$ has $27/lcm(9,3)=3$ cycle graphs with its graph length $|G^+_{\{9,3\}}|=lcm(9,3)=9$. The dashed line denotes a single period of $C_9^+$ and the dotted line represents a single period of $C_3^+$. }}
    \label{fig:cyclegraphdecompositions}
\end{figure}

%In this paper, 

With our purpose, we resume our extension of construction \ref{thm:graph-torus} by considering a subset of the edges (\ref{eq:cycle-edges}) so that we can obtain disjoint graph unions of cycle graphs. %In other words, we work on a special case of a Cartesian product of two dihedral symmetries, $D_+\times D_- =\cup $, instead of Cartesian products of $ab$ dihedral symmetries, $D_{\alpha_1}\times \cdots \times D_{\alpha_a} \times D_{\beta_1} \times \cdots \times D_{\beta_b}$ to simplify our discussions. Here, $D_+$ and $D_-$ are the dihedral symmetries associated with the set of $+$-regions and $-$-regions. Hence, 
We consider subgraphs $G^\pm_{\{\alpha\}}$ defined as 
\begin{equation}\label{eq:multicyclegraphsalpha}
    G^\pm_{\{\alpha\}} : = (V^\pm_{\{\alpha\}},E^\pm_{\{\alpha\}})
\end{equation}
where 
\begin{equation}\label{eq:multicyclegraphsalpha-edge}
\begin{split}
    E^\pm_{\{\alpha\}} := \big\{\big{\langle}(a^\pm_{i_1},\cdots,a^\pm_{i_{n_\alpha}}),(a^\pm_{i'_1},\cdots,a^\pm_{i'_{n_\alpha}})\big{\rangle}&|(i_1-i'_1,\cdots,i_{n_\alpha}-i'_{n_\alpha}) =( \frac{\alpha_1-1}{2}, \cdots , \frac{\alpha_a-1}{2})\\
    &\text{ or }(i_1-i'_1,\cdots,i_{n_\alpha}-i'_{n_\alpha}) =(- \frac{\alpha_1-1}{2}, \cdots ,- \frac{\alpha_a-1}{2})\big\},\\
\end{split}
\end{equation}
and
\begin{equation}
    \{\alpha\}:=\{\alpha_1,\cdots, \alpha_{n_\alpha}\}.
\end{equation}

In parallel, we define 
\begin{equation}\label{eq:multicyclegraphsbeta}
    G^\pm_{\{\beta\}} : = (V^\pm_{\{\beta\}},E^\pm_{\{\beta\}})
\end{equation}
where
\begin{equation}
    \{\beta\}:=\{\beta_1,\cdots, \beta_{n_\beta}\}.
\end{equation}

Now, we state that the subgraphs $G^\pm_{\{\alpha\}}, G^\pm_{\{\beta\}}$ decompose into disjoint cycle subgraphs. %Moreover, the disjoint cycle subgraphs of $G^\pm_{\{\alpha\}}$ and $G^\pm_{\{\beta\}}$ span $G^\pm_{\{\alpha\}}$ and $G^\pm_{\{\beta\}}$, respectively.
Moreover, $G^\pm_{\{\alpha\}}$ and $G^\pm_{\{\beta\}}$ are spanned by their cycle subgraphs respectively, see figure \ref{fig:cyclegraphdecompositions} for examples. 

%\JN{Figure missing here.}

\begin{lemma}[Decomposition with cycle graphs]\label{lem:decomposition-with-cyclegraphs}

    \begin{equation}
        G^\pm_{\{\alpha\}} = \underset{\kappa_{\{\alpha\}}=1,\cdots,m}{\bigcup} C^\pm_{\kappa_{\{\alpha\}}}, \; G^\pm_{\{\beta\}} =\underset{\kappa_{\{\beta\}}=1,\cdots,n}{\bigcup} C^\pm_{\kappa_{\{\beta\}}}
    \end{equation}
    where $C^\pm_{\kappa_{\{\alpha\}}}:=(V^\pm_{\kappa_{\{\alpha\}}},E^\pm_{\kappa_{\{\alpha\}}})$ are cycle graphs with $|V^\pm_{\kappa_{\{\alpha\}}}|=|E^\pm_{\kappa_{\{\alpha\}}}|=lcm(\alpha_1,\cdots,\alpha_{n_\alpha})$ for $\forall \kappa_{\{\alpha\}}$, and $m=\prod_{s}\alpha_s/lcm(\alpha_1,\cdots,\alpha_{n_\alpha}) $. Similarly, $C^\pm_{\kappa_{\{\beta\}}}:=(V^\pm_{\kappa_{\{\beta\}}},E^\pm_{\kappa_{\{\beta\}}})$ are cycle graphs with $|V^\pm_{\kappa_{\{\beta\}}}|=|E^\pm_{\kappa_{\{\beta\}}}|=lcm(\beta_1,\cdots,\beta_{n_\beta})$ for $\forall \kappa_{\{\beta\}}$, and $n=\prod_{t}\beta_t/lcm(\beta_1,\cdots,\beta_{n_\beta})$. Here, $lcm$ is the least common multiple. Note that $\bigcup$\footnote{We denote $\bigcup$ both union of sets and graph unions unless there is any confusion.}is a graph union, e.g., 
    \begin{equation}
        \underset{\kappa_{\{\alpha\}}=1,\cdots,m}{\bigcup} C^\pm_{\kappa_{\{\alpha\}}} = \Big(\underset{\kappa_{\{\alpha\}}=1,\cdots,m}{\bigcup}V^\pm_{\kappa_{\{\alpha\}}},\underset{\kappa_{\{\alpha\}}=1,\cdots,m}{\bigcup}E^\pm_{\kappa_{\{\alpha\}}}\Big).
    \end{equation}

    Furthermore, the disjoint unions of $C^\pm_{\kappa_{\{\alpha\}}}$ span $G^\pm_{\{\alpha\}}$. Similarly, the disjoint unions of $C^\pm_{\kappa_{\{\beta\}}}$ span $G^\pm_{\{\beta\}}$.  %$G^\pm_{\{\beta\}}$.

\end{lemma}
\begin{proof}

    From (\ref{eq:multicyclegraphsalpha}) and (\ref{eq:multicyclegraphsalpha-edge}), each vertex has degree $2$. Recall from (\ref{eq:mod}) that $a^{\pm}_{i_s} \equiv a^{\pm}_{i_s+\alpha_s}$\footnote{Here, we use the simplified notation, i.e., $a^{\pm}_{(i_s,s)} \equiv a^{\pm}_{(i_s+\alpha_s,s)} \to a^{\pm}_{i_s} \equiv a^{\pm}_{i_s+\alpha_s}$. } for $\forall i_s$.
    
    When $n_\alpha=2$, $G^{\pm}_{\{\alpha\}}$ is a subgraph of $C^{\pm}_{\alpha_1} \Box C^{\pm}_{\alpha_2}$. 
    %\begin{equation}
    %    G^{\pm}_{\{\alpha\}} = C^{\pm}_{\alpha_1} \Box C^{\pm}_{\alpha_2}.
    %\end{equation}
    Each vertex of $G^{\pm}_{\{\alpha\}}$ is denoted as $(a^\pm_{i_1},a^\pm_{i_{2}})$ for $\forall i_{1}= 1, \cdots, \alpha_1$ and $\forall i_{2} = 1, \cdots, \alpha_2$. Suppose $\alpha_1 > \alpha_2$. For fixed $i_{1}$ and $i_{2}$, we obtain the cycle graph whose edges connect the vertices shifted by, for instance, $(\frac{\alpha_1-1}{2},\frac{\alpha_2-1}{2})$, until
    \begin{equation}
        (a^\pm_{i_{1}+\alpha_1\alpha_2},a^\pm_{i_{2}+\alpha_1\alpha_2}) \equiv(a^\pm_{i_{1}},a^\pm_{i_{2}}).
    \end{equation}
    Thus, the graph length is at most the modulus $\alpha_1\alpha_2$, i.e., $|G^{\pm}_{\{\alpha\}}| =\alpha_1\alpha_2$.
    
    If $\alpha_1 \text{ mod } \alpha_2 = 0$, or $\alpha_1 = \xi \alpha_2$ for $\xi \in \mbN$, then, for a fixed $i_{1}$ and $i_{2}$, the smallest length of cycle graph is $\alpha_1$ because
    \begin{equation}
        (a^\pm_{i_{1}+\alpha_1},a^\pm_{i_{2}+ \xi \alpha_2}) =(a^\pm_{i_{1}+\alpha_1},a^\pm_{i_{2}+\alpha_1}) \equiv (a^\pm_{i_{1}},a^\pm_{i_{2}}).
    \end{equation}
    Moreover, for a fixed $i_{1}$, there are $\alpha_2$ distinct initial vertices $(i_{1},i_{2})$ because there are only $\xi$ distinct vertices,
    \begin{equation}
        (a^\pm_{i_{1}},a^\pm_{i_{2}}),(a^\pm_{i_{1}},a^\pm_{i_{2}}),\cdots, (a^\pm_{i_{1}+(\xi-1)\alpha_2}, a^\pm_{i_{2}}),
    \end{equation}
    in a single cycle graph for fixed $i_{2}$.
    Hence, there are $m= \alpha_2$ cycle graphs $C^{\pm}_{\kappa_{\{\alpha\}}}=(V^\pm_{\kappa_{\{\alpha\}}},E^\pm_{\kappa_{\{\alpha\}}})$ with length $|C^{\pm}_{\kappa_{\{\alpha\}}}|=\alpha_1$. 
    
    This implies that 
    \begin{equation}
        \sum_{\kappa_{\{\alpha\}}}|V^\pm_{\kappa_{\{\alpha\}}}|  = |V^\pm_{\{\alpha\}}| =\alpha_1\alpha_2.
    \end{equation}
    Therefore, the disjoint union of $C^{\pm}_{\kappa_{\{\alpha\}}}$ spans $G^{\pm}_{\{\alpha\}}$.

    If $\alpha_1 \text{ mod } \alpha_2 \neq 0 $, the smallest length of the cycle graph is $\alpha_1\alpha_2$. There is a $m=1$ cycle graph $C^{\pm}_1=(V^\pm_{\kappa_{\{\alpha\}}=1},E^\pm_{\kappa_{\{\alpha\}}=1})$ which satisfies
    \begin{equation}
        |V^\pm_1| =  |V^\pm_{\{\alpha\}}|=\alpha_1\alpha_2.
    \end{equation}
    In this case, $(a^\pm_{i_1},a^\pm_{i_2})$ for $\forall i_{1}$ and $\forall i_{2}$ are the vertices of a single cycle graph. $C^{\pm}_1$ trivially spans $G^{\pm}_{\{\alpha\}}$. 
    
    In short, the length of each disjoint cycle and the total number of disjoint cycles are determined by $lcm(\alpha_1,\alpha_2)$ and $\alpha_1\alpha_2/lcm(\alpha_1,\alpha_2)$, respectively. In general\footnote{We omit the general case for any $n_\alpha$ as the extension is trivial.}, the length of each disjoint cycle is given by $lcm(\alpha_1,\cdots,\alpha_{n_\alpha})$. The total number of disjoint cycles is $(\prod_s \alpha_s)/lcm(\alpha_1,\cdots,\alpha_{n_\alpha})$.

    %The set of edges decomposes into disconnected edges by the definition of $G^\pm_{\{\alpha\}}$. %That is, 

    %$G^\pm_{\{\alpha\}_c}$ for $\forall c$ spans %$G^\pm_{\{\alpha\}}$ because
    %\begin{equation}
    %    \underset{c=1,\cdots,m}{\bigcup} V_{\{\alpha\}_c} = V_{\{\alpha\}}
    %\end{equation}
    
\end{proof}
With the graph Cartesian product in definition \ref{def:complementgraphproduct}, we construct
\begin{equation}
    G_\mL = G^+_{\{\alpha\}} \Box G^-_{\{\beta\}}, \; G_\mR = G^-_{\{\alpha\}} \Box G^-_{\{\beta\}}
\end{equation}
From lemma \ref{lem:decomposition-with-cyclegraphs}, 
\begin{equation}
    G_\mL = \underset{\kappa_{\{\alpha\}}=1,\cdots,m,\kappa_{\{\beta\}}=1,\cdots,n}{\bigcup}  C^+_{\kappa_{\{\alpha\}}} \Box C^-_{\kappa_{\{\beta\}}}.
\end{equation}
$C^+_{\kappa_{\{\alpha\}}} \Box C^-_{\kappa_{\{\beta\}}}$ for each pair $(\kappa_{\{\alpha\}},\kappa_{\{\beta\}})$ is embeddable on a $2$-torus as in lemma \ref{lem:embed-torus}, i.e.,
\begin{equation}
     C^+_{\kappa_{\{\alpha\}}} \Box C^-_{\kappa_{\{\beta\}}} \to T_{\mL_\tau} ,\;  C^+_{\kappa_{\{\alpha\}}} \Box C^-_{\kappa_{\{\beta\}}} \to T_{\mR_\tau}
\end{equation}
where $\tau := (\kappa_{\{\alpha\}},\kappa_{\{\beta\}})$. Thus, we write
\begin{equation}\label{eq:leftrighttorus-decomposition}
    T_\mL = \underset{\tau}{\bigcup} \;T_{\mL_\tau} ,\; T_\mR = \underset{\tau}{\bigcup} \;T_{\mR_\tau}
\end{equation}

There exists the dual graph for each $\tau$, i.e.
\begin{equation}
    T_{\mR_\tau} = \hat{T}_{\mL_\tau}.
\end{equation}
Therefore,
\begin{equation} \label{eq:decompositions-torii}
    T_{\mR} = \hat{T}_{\mL} = \underset{\tau}{\bigcup} \;\hat{T}_{\mL_\tau} = \underset{\tau}{\bigcup} \;T_{\mR_\tau}   .
\end{equation}

Based on the constructions, we summarize the recipe of graphical representation of $(\alpha_1,\cdots,\alpha_{n_\alpha};\beta_1,\cdots,\beta_{n_\beta})$-conjectures as follows.
\begin{construction}[Geometrizing the generalized toric conjectures\footnote{%One should note that the graphs constructed from Construction \ref{thm:tile-general} and \ref{thm:graph-general} are different. The latter ones are decomposed into $2$-torii as opposed to the former ones, where such decomposition is not necessarily visible.
Note that the graphs constructed from Construction \ref{thm:tile-general} and \ref{thm:graph-general} are different in the sense that the latter ones are explicitly constructed from unions of $2$-torii. The unit cells in Construction \ref{thm:tile-general} do not necessarily admit such decomposition, see figure \ref{fig:21-tilingcell}.\label{ft:construction4-2}}: Graph theoretical method] \label{thm:graph-general}\

    Consider $(\alpha_1,\cdots,\alpha_{n_\alpha};\beta_1,\cdots,\beta_{n_\beta})$-conjectures.
    \begin{enumerate}
        \item Construct the graphs, $G^+_{\{\alpha\}}$ and $G^-_{\{\beta\}}$, and find their decompositions of cycle subgraphs. 
        \item Construct $G_\mL = G^+_{\{\alpha\}} \Box G^-_{\{\beta\}}$. %= \underset{\tau}{\cup}\; G_{\mL_\tau}$ for $\tau=1,\cdots,mn$. %$G_\mR=\hat{G}_\mL$ by taking graph dual of $G_\mL = G^+_\alpha \Box G^-_\beta$.
        %\item Construct $G_\mL = G^+_\alpha \Box G^-_\beta$ and $G_\mR = G^-_\alpha \Box G^-_\beta$.
        \item Obtain $T_\mL = \underset{\tau}{\cup} \; T_{\mL_\tau}$ by embedding $G_\mL$ on disjoint torii.
        %\footnote{This construction is equivalent to the one discussed in footnote \ref{clarify-rhombi}, as we are considering unions of disjoint $2$-torii instead of a more general higher dimensional torus.}
        \item $T_\mR = \hat{T}_\mL=\underset{\tau}{\cup} \; \hat{T}_{\mL_\tau} =\underset{\tau}{\cup} \; T_{\mR_\tau} $ 
    \end{enumerate}

    Note that the graph explicitly used for the proof by a contraction map in this paper corresponds to $T_\mR$.

\end{construction}

\subsection{Proof methods}

For the generalized toric conjectures, we denote the bitstrings of the LHS and RHS as $X$ and $Y$, respectively. Similar to the case of $(\alpha,\beta)$-inequalities, we define the occurrence bitstrings of, for instance, $a_i$ as
\begin{equation}
(X^{a_i})_u=
\begin{cases}
1 & \text{if $a_i \subseteq L_u$}\\
0 & \text{otherwise}\\
\end{cases}, \;
(Y^{a_i})_v=
\begin{cases}
1 & \text{if $a_i \subseteq R_v$ }\\
0 & \text{otherwise}\\
\end{cases}.
\end{equation}

%To find a geometric contraction map of the generalized toric conjectures on $\Gamma(\alpha_1,\cdots,\alpha_a;\beta_1,\cdots,\beta_b)$, 
%We geometrically assign the bitstrings $X$ and $Y$ on each torus $T_{\mR_\tau}$ due to the decomposition $T_\mR =  \underset{\tau}{\bigcup} \;T_{\mR_\tau}$ in (\ref{eq:decompositions-torii}). 

\begin{figure}[t]
    \centering
    \includegraphics[width=1.0\linewidth]{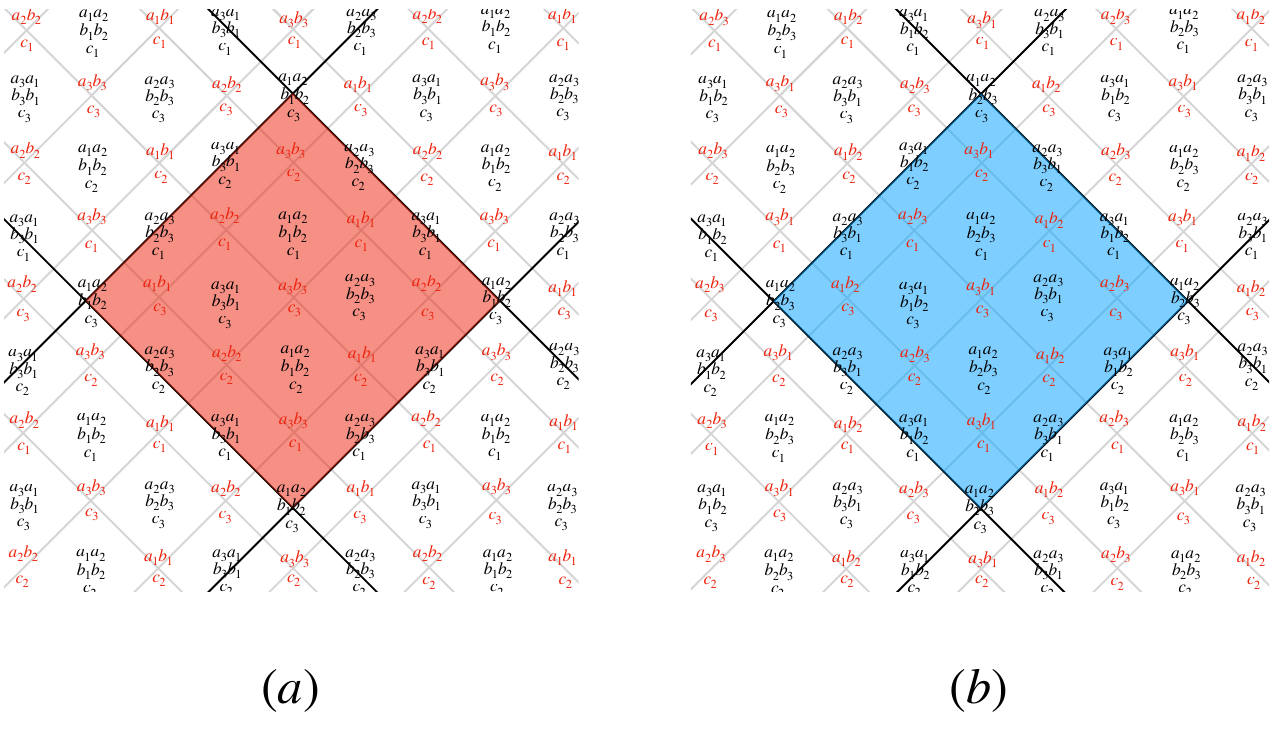}
    \caption{\small{(a), (b), and (c) represent the fundamental domain of disjoint $2$-torii in (\ref{eq:color-conjecture}). See (c) in figure \ref{fig:green}. Each color in the figures matches the colors in (\ref{eq:color-conjecture}).}}
    \label{fig:333decomposition}
\end{figure}

\begin{figure}[t]
    \centering
    \includegraphics[width=0.5\linewidth]{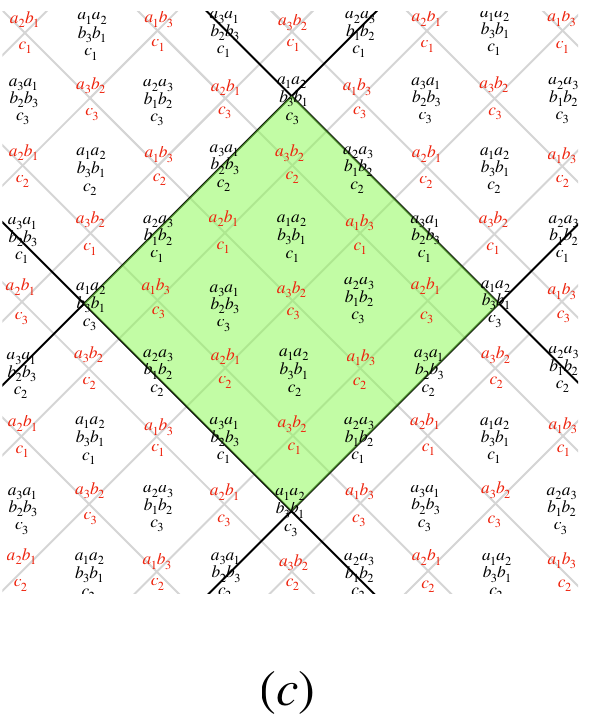}
    \caption{\small{(a), (b), and (c) represent the fundamental domain of disjoint $2$-torii in (\ref{eq:color-conjecture}). See (a) and (b) in figure \ref{fig:333decomposition}. Each color in the figures matches the colors in (\ref{eq:color-conjecture}).}}
    \label{fig:green}
\end{figure}

The decomposition $T_\mR =  \underset{\tau}{\bigcup} \;T_{\mR_\tau}$ in (\ref{eq:decompositions-torii}) implies that the geometric terms in %$(\alpha_1,\cdots,\alpha_a;\beta_1,\cdots,\beta_b)$-
the generalized toric conjectures can be decomposed into the group of terms corresponding to $2$-torii. For example, $(3,3;3)$-conjecture has three disjoint $2$-torii in the geometric part, i.e.\footnote{(Color online) the two terms in (\ref{eq:color-conjecture}) with the same color correspond to a single $2$-torus.}, see figure \ref{fig:333decomposition},
\begin{equation}\label{eq:color-conjecture}
\begin{matrix}
    \mathcolor{red}{\sum_{i=1}^{3}\sum_{j=1}^{3} S_{a_{i}^+b_{i}^+c^-_{j} }}\\
    +\mathcolor{blue}{\sum_{i=1}^{3}\sum_{j=1}^{3} S_{a_{i}^+b_{i+1}^+c^-_{j} }}\\
    +\mathcolor{teal}{\sum_{i=1}^{3}\sum_{j=1}^{3} S_{a_{i}^+b_{i+2}^+c^-_{j} }}\\
\end{matrix}
\geq
\begin{matrix}
    \mathcolor{red}{\sum_{i=1}^{3}\sum_{j=1}^{3}  S_{a_{i}^-b_{i}^-c^-_{j}}} \mathcolor{white}{+\text{``non-geometric terms"}}\\
    +\mathcolor{blue}{\sum_{i=1}^{3} \sum_{j=1}^{3}S_{a_{i}^-b_{i+1}^-c^-_{j}}}+\text{``non-geometric terms"}\\
    +\mathcolor{teal}{\sum_{i=1}^{3}\sum_{j=1}^{3} S_{a_{i}^-b_{i+2}^-c^-_{j}}} \mathcolor{white}{+\text{``non-geometric terms"}}\\
\end{matrix}.
\end{equation}

As a result, the bitstrings on the LHS, $X$, decomposes into the bitstrings $x$ on each torus, i.e.,
\begin{equation}\label{eq:Xdecomposition}
    X=x_{(1,1)}\times \cdots \times x_{\tau=(\kappa_{\{\alpha\}},\kappa_{\{\beta\}})} \times \cdots \times x_{(m,n)}.
\end{equation}
This simplifies the proof of the generalized toric conjectures because the geometric assignment of bitstrings $X$ and $Y$ reduces to the one of $x_{\tau}$\footnote{When $G_\mL$ and $G_\mR$ are constructed with the graphs in (\ref{eq:higher-dim-toroidal-graph}) instead of the subgraphs $G^\pm_{\{\alpha\}}$ and $G^\pm_{\{\beta\}}$, the geometry of bitstrings becomes a set of higher dimensional objects rather than the knots we had in section \ref{sec:review-toric}. We briefly discuss this point in section \ref{sec:discussions}.}.

%This simplifies the geometric assignment of bitstrings $X$ and $Y$ since we can use the same assignment on each torus as the one in section \ref{sec:review-toric}\footnote{In general, the geometry of bitstrings becomes a set of higher dimensional objects rather than knots we had in section \ref{sec:review-toric}. We briefly discuss this point in section\ref{sec:discussions}.}.

Let $F$ be a candidate geometric contraction map given by
\begin{equation}
    F(X):= \tilde{F}(X) \times \tilde{F}_{ng}
\end{equation}
where $\times$ is a cartesian product. $\tilde{F}(X) \in \{0,1\}^{l}$ has the information of bitstrings on the set of vertices $V_\mR=\mR$, and $\tilde{F}_{ng}\in \{0,1\}^{r-l}$ contains the information of non-geometric terms. %The rules, such as \ref{rul:toric}, determine the non-geometric bits in $\tilde{F}_{global}$.
One should note that $\tilde{F}(X)$ does not necessarily have the following decomposition,
\begin{equation}\label{eq:decomposition-F}
    \tilde{F}(x_{(1,1)}\times \cdots \times x_{\tau} \times  \cdots \times x_{(m,n)}) = \tilde{f} (x_{(1,1)})\times \cdots \times  \tilde{f}(x_{\tau}) \times \cdots \times \tilde{f}(x_{(m,n)})
\end{equation}
where $\tilde{f}$ is a part of the geometric contraction map introduced in (\ref{eq:contractionmap-toric}). The local action of $\tilde{f}$ on $x_\tau$ is defined by Rule \ref{rul:toric}. The decomposition does not happen when the bit-flip of $(x)_u$ on $2$-torus results in double bit-flips on two vertices\footnote{In such a case, we say that there is an \textit{interaction} among the $2$-torii. We comment on it in section \ref{sec:discussions}.}. One vertex is on the torus where $(x)_u$ lives. The other vertex is on the different torus.

This paper focuses on the case when we have the decomposition,
\begin{equation}
    F(X)= \tilde{f}(x_{(1,1)})\times \cdots \times \tilde{f}(x_{\tau}) \times \cdots \times \tilde{f}(x_{(m,n)} )\times \tilde{F}_{ng}.
\end{equation}

Here, we give the rules that define the geometric contraction map for the generalized toric conjectures.
\begin{rules}\label{rul:general-toric}\

\noindent
\underline{1. Rules on geometric bitstrings $\tilde{F}(X)$, vertices}

    From the decomposition, 
    \begin{equation}
        \tilde{F}(x_{(1,1)}\times \cdots \times x_{\tau} \times \cdots \times x_{(m,n)}) = \tilde{f} (x_{(1,1)})\times \cdots \times \tilde{f}(x_{\tau}) \times  \cdots \times \tilde{f}(x_{(m,n)}),
    \end{equation}
    we apply the rule \ref{rule1} to $\tilde{f}(x_\tau)$ for every $x_\tau$.

\noindent
\underline{2. Rules on the non-geometric bitstrings $\tilde{F}_{ng}$}

    The non-geometric bitstrings $\tilde{F}_{ng}$ are constrained\footnote{In the case of a single non-geometric term, the non-geometric bitstring is uniquely determined by (\ref{eq:general-toricrule2}). However, for multiple non-geometric terms, there could be some residual degeneracy after the condition (\ref{eq:general-toricrule2}). In this case, one needs to check if it is compatible with theorem \ref{thm:proofbycontraction}. We make detailed comments in section \ref{sec:examples}. } by
    \begin{equation} \label{eq:general-toricrule2}
        \|X\|_1=\|F(X)\|_1\ \text{mod}\ 2.
    \end{equation}
    
    %\begin{enumerate}
    %    \item[a)] 
    %    \begin{equation} \label{eq:general-toricrule2}
    %        \|X\|_1=\|F(X)\|_1\ \text{mod}\ 2.
    %    \end{equation}
    %   \item[b)]
    %    Boundary conditions, i.e., 
    %    \begin{equation}
    %        F(X^{a_{(i_s,s)}}) =Y^{a_{(i_s,s)}},\;F(X^{b_{(j_t,t)}}) =Y^{b_{(j_t,t)}}
    %    \end{equation}
    %    for $\forall i_s=1,\cdots,\alpha_s$, $\forall j_t=1,\cdots,\beta_t$, $\forall s=1,\cdots,a$, and $\forall t=1,\cdots,b$.
    %\end{enumerate}

\end{rules}

We briefly discuss a few examples in section \ref{sec:examples}. %Note that Rule \ref{rul:toric} did not explicitly state the boundary conditions because the mod condition uniquely fixes the $\tilde{f}$. We comment further on this in \ref{sec:examples}.  %However, this does not mean the action of geometric contraction map $F:X \to F(X)$ can decompose on each torus. This depends on the global term of the conjectures.
%F:X \to f(x)\times \cdots f(x)\times f_g where f_g=\{0,1\}^{r-l} is determined by the global terms
The non-geometric terms constrain the global entanglement structure by adding the entanglement entropy between the sets of regions, $A_1,\cdots,A_{n_\alpha}$ and $B_1,\cdots,B_{n_\alpha}$. For example, the last term of $(\alpha;\beta)$-inequalities is $S_{A}$. %The last term ties all the terms up and makes some inequalities non-redundant. 

The balance and superbalance conditions partially or entirely determine the non-geometric terms of generalized toric conjectures. Hence, we discuss the balance and superbalance conjectures in the following subsection before we move on to the examples.

\subsection{Balance and Superbalance}\label{sec:balance-superbalance}

We begin this section by briefly reviewing the definitions of balance and superbalance of HEIs.
\begin{definition}[Balance] \cite{Bao:2015bfa,HEarrangement2019Hubeyetal}
    A HEI (\ref{eq:genentineq}) is balanced if all singleton characters appear an equal number of times on both sides.
\end{definition}

\begin{definition}[Superbalance] \cite{HEarrangement2019Hubeyetal,HErepackaged2019He,superbalance2020He,Hernandez-Cuenca:2023iqh}
    A HEI (\ref{eq:genentineq}) is superbalanced if the inequality, under all the permutations of labels of regions, including purification is balanced.
\end{definition}

All $(\alpha;\beta)$-inequalities can be expressed by conditional entropies in the form
\begin{equation} \label{eq:ce}
    \sum_{i=1}^{\alpha}\sum_{j=1}^{\beta} S(a_{i+\frac{\alpha-1}{2}}|a^-_{i} b^-_j )\geq  S_{A}.
\end{equation}
(\ref{eq:ce}) is neither balanced nor superbalanced when the non-geometric term, $S_A$, is absent. It is by inclusion of the non-geometric term $S_A$ on the RHS of (\ref{eq:ce}) that makes the inequalities balanced and superbalanced. In this section, we will modify the non-geometric terms in the RHS of (\ref{eq:ce}) to conjecture balanced and superbalanced HEI candidates. In particular, we will modify the non-geometric terms such that they are the sum of entanglement entropies of all possible non-redundant combinations among the sets $A_1,\cdots, A_a$ and $B_1,\cdots, B_b$. Any entanglement entropy related to another by the purification symmetry does not appear in the non-geometric terms of the conjectures.

In general, we write $(\alpha_1,\cdots,\alpha_{n_\alpha};\beta_1,\cdots,\beta_{n_\beta})$-conjectures as
\begin{equation}
\begin{split}
    \sum_{i_1,\cdots, i_{n_\alpha}=1}^{\alpha_1,\cdots,\alpha_{n_\alpha}} &\sum_{j_1,\cdots, j_b=1}^{\beta_1,\cdots,\beta_{n_\beta}}
     S(a_{i_1+\frac{\alpha-1}{2}}\cdots a_{i_{n_\alpha}+\frac{\alpha-1}{2}} |a^-_{i_1}\cdots a^-_{i_{n_\alpha}} b^-_{j_1}\cdots b^-_{j_{n_\alpha}})\\
    & \geq \sum_{s=1}^{n_\alpha}\chi^{A_s} S_{A_s} + \sum_{t=1}^b\chi^{B_t} S_{B_t} \\
    & + \sum_{s,s'=1}^{n_\alpha}\chi^{A_{s}A_{s'}} S_{A_{s}A_{s'}} + \sum_{t,t'=1}^{n_\beta}\chi^{B_{t}B_{t'}} S_{B_tB_{t'}} + \sum_{s=1}^{n_\alpha}\sum_{t=1}^{n_\beta}\chi^{A_{s}B_{t}} S_{A_sB_t}\\
    &+\cdots
\end{split}
\end{equation}
where $\chi$'s are the integer coefficients. Then, we fix the parameters $\chi$ such that the conjectures are balanced or superbalanced.

%\begin{definition}[Balance]
%    single regions appear the same number of times
%\end{definition}
%Note that the balance condition depends on the purification, see the discussion in \cite{}superbalance.

%\begin{definition}[Superbalance]

%    An inequality is superbalance if it satisfies one of the following conditions
%    \begin{enumerate}
    
%        \item pairs of regions appear the same number of times
        
%        \item the inequality is balanced by interchanging the purifier to any other regions.
        
%    \end{enumerate}
    
%\end{definition}

%For $(\alpha_1,\cdots,\alpha_{n_\alpha};\beta_1,\cdots,\beta_{n_\beta})$-conjectures, total number of global terms with all combinations among $A_1,\cdots,A_{n_\alpha}, B_1,\cdots, B_{n_\beta}$ can be reduced by the purification symmetry.
%\JN{JN: It feels like a discontinuity here. Maybe this part can be rephrased? Keeping it for my to-do.}

The total number $\sigma_p$ of non-redundant parameters in the non-geometric terms is $\sigma_p=2^{{n_\alpha}+{n_\beta}-1}-1$. The balance conditions fix $\sigma_b =n_\alpha+n_\beta-1$ number of parameters. The superbalance conditions determine $\sigma_{sb}=\binom{n_\alpha+n_\beta}{2}$\footnote{ Note that $\sigma_{sb}= \binom{n_\alpha+n_\beta-1}{2}+ \sigma_b$, which contains $\sigma_b$, because the superbalance conditions imply the balance conditions.} number of parameters.

For example, the potential non-geometric terms of $(\alpha;\beta)$-inequalities are $S_A,S_B,S_{AB}$. Recall that $S_{AB}=0$ and $S_A=S_B$. Without loss of generality, we can write
\begin{equation} \label{eq:toric-parameter}
    \sum_{i=1}^{\alpha}\sum_{j=1}^{\beta} S(a_{i+\frac{\alpha-1}{2}}|a^-_{i} b^-_j )\geq  \chi^A S_{A}
\end{equation}
with the coefficient $\chi^A$. In addition, we have $n_\alpha + n_\beta =2$, $\sigma_p=\sigma_b=\sigma_{sb}=1$. This implies that the parameter $\chi^A$ in (\ref{eq:toric-parameter}) is fully fixed by both balance and superbalance conditions, i.e., $\chi^A=1$.

The $(\alpha_1,\alpha_2;\beta)$-conjectures for ${n_\alpha}+{n_\beta}=3$ can be expressed as
\begin{equation}\label{eq:2p1m}
\sum_{i,j=1}^{\alpha_1,\alpha_2}\sum_{j=1}^{\beta} S(a_{i+\frac{\alpha_1-1}{2}}b_{j+\frac{\alpha_2-1}{2}} |a_i^-b_j^-c_k^-)\geq \chi^A S_{A} +\chi^B S_{B}+ \chi^C S_{C}.
\end{equation}
where $\chi^A,\chi^B,\chi^C\in \mbZ$. We have $\sigma_b=2$, $\sigma_{sb}=3$, and $\sigma_p=3$. Hence, all the parameters $\chi^A,\chi^B,\chi^C$ of the non-geometric terms of superbalanced conjectures are fixed, whereas those of balanced conjectures have a single free parameter since $\sigma_p-\sigma_b=1$. In general, for ${n_\alpha}+{n_\beta}>3$, the non-geometric terms of superbalanced conjectures have $(\sigma_p-\sigma_{sb})$ numbers of free parameters.

%Below we present the balanced and superbalanced $(\alpha_1,\alpha_2;\beta)$-conjectures as examples. We denote the conditional entropies on the LHS as `C.E.' just for brevity.

%In general, $(\alpha;\beta)$-inequalities can be written as (\ref{eq:toric-parameter}).
%\begin{equation}
 %   \sum_{i=1}^{\alpha}\sum_{j=1}^{\beta} S(a_{i+\frac{\alpha-1}{2}}|a_i^-b_j^-)\geq \chi S_{A}
%end{equation}
%where $\chi \in \mbZ$. 
%As discussed above, we have $\sigma_p=\sigma_b=\sigma_s=1$. This implies, we can only have superbalanced $(\alpha;\beta)$-inequalities, as the balance and superbalance conditions, both specify $\chi^A=1$. %Note that we used the fact that $S_{A^{(\alpha)}_1}=S_{B^{(\beta)}_1}$ and $S_{A^{(\alpha)}_1B^{(\beta)}_1} =0$ since $A\cup B$ is a pure system. When the purifier $O$ is a part $A$ or $B$, both the balance conditions determine $\chi=1$. Hence, we only have superbalance $(\alpha;\beta)$-inequalities.
%\JN{JN: There is repeatation of information from previous page in the paragraph above.}
%The $(\alpha_1,\alpha_2;\beta)$-conjectures can be expressed as
%\begin{equation}\label{eq:2p1m}
%\sum_{i,j=1}^{\alpha_1,\alpha_2}\sum_{j=1}^{\beta} S(a_{i+\frac{\alpha_1-1}{2}}b_{j+\frac{\alpha_2-1}{2}} |a_i^-b_j^-c_k^-)\geq \chi^A S_{A} +\chi^B S_{B}+ \chi^C S_{C}.
%\end{equation}
%where $\chi^A,\chi^B,\chi^C\in \mbZ$.

Below, we denote the conditional entropies on the LHS as `C.E.' just for brevity. By the balance conditions,
\begin{enumerate}
    \item[i)] when the purifier $O \in A$, 
    \begin{equation}
        \chi^A + \chi^B = 0,\;\chi^A + \chi^C = \alpha_2,
    \end{equation}
    thence \begin{equation}\label{eq:balancepurifierinA}
    C.E. \geq   (\alpha_2-\chi^C) S_{A} -(\alpha_2-\chi^C) S_{B}+ \chi^C S_{C}
    \end{equation}

    \item[ii)] when $O \in B$, 
    \begin{equation}
        \chi^A + \chi^B = 0,\; \chi^B+\chi^C = \alpha_1,
    \end{equation}
    thence 
    \begin{equation}
        C.E. \geq  -  (\alpha_1-\chi^C) S_{A} +(\alpha_1-\chi^C) S_{B}+ \chi^C S_{C}.
    \end{equation}

    \item[iii)] when $O \in C$, 
    \begin{equation}
        \chi^A + \chi^C = \alpha_2,\; \chi^B+\chi^C = \alpha_1
    \end{equation}
    thence 
    \begin{equation}
        C.E. \geq   (\alpha_2-\chi^C) S_{A} +(\alpha_1-\chi^C) S_{B}+ \chi^C S_{C}.
    \end{equation}
    
\end{enumerate}
Note that the non-geometric terms of each conjecture contain $\chi^C$ as a free parameter because $\sigma_p-\sigma_s=1$. %In other words, only two independent equations exist for the three parameters.

For $(\alpha_1,\alpha_2;\beta)$-conjectures to be superbalanced, the parameters need to satisfy essentially
\begin{equation}
    \chi^A + \chi^B = 0,\; \chi^B+\chi^C = \alpha_1, \; \chi^A + \chi^C = \alpha_2.
\end{equation}
Thus, we have
\begin{equation} \label{eq:2p1msuperbalance}
   C.E. \geq -  \frac{\alpha_1-\alpha_2}{2} S_{A} +\frac{\alpha_1-\alpha_2}{2} S_{B}+ \frac{\alpha_1+\alpha_2}{2} S_{C}.\\
\end{equation}

%\JN{JN: Does eq. \ref{eq:2p1msuperbalance} always gives a true inequality? If not, I am not sure, how good is the idea to write it here. At least, we should clarify about validity here.}

%For $(\alpha_1,\alpha_2;\beta_1,\beta_2)$-inequalities to be balanced
%\begin{equation}
%    aaa
%\end{equation}

%For $(\alpha_1,\alpha_2;\beta_1,\beta_2)$-inequalities to be superbalanced
%\begin{equation}
%    aaa
%\end{equation}

%meaning of tuning the parameter, e.g., $\chi^C$ $=>$ modifying with $\pm I_2$

%lemma: for higher, i guess $I_3,\cdots, I_n$

%toric inequality is modified by $I_1 =S_A$

%higher has more degrees of freedom, more room to find superbalance inequality

%lemma: number of parameter increases as ...

%the generalized toric conjectures that cannot be written only with the sum of toric inequalities. We call them the non-redundant conjectures\footnote{One should not be confused with ``primitive'' inequalities.}.
 %In short, we need a hybrid contraction map that geometrically determines the geometric bitstrings $\tilde{F}(X')$ and numerically determines the non-geometric bitstrings $\tilde{F}_{global}$. 

\section{Examples of Generalized Toric Inequalities} \label{sec:examples}

This section presents examples of valid inequalities found from the conjectures. There are redundant inequalities implied only by the toric inequalities and the ones that are not. We first describe that the former inequalities can be expressed as a sum of $(\alpha;\beta)$ inequalities, and thus they are redundant HEIs. Then, we present one of the latter inequalities as an example, which is balanced. We leave the further search of the facets of HEC from the conjectures\footnote{One can also perform a tightening procedure described in \cite{Hernandez-Cuenca:2023iqh} by adding conditional tri-partite information quantities starting from our inequalities. We leave this as a future exercise.} for future work as discussed in section \ref{sec:discussions}. 

%The redundant ones can be expressed as a sum of $(\alpha;\beta)$ inequalities. The non-redundant inequalities are the ones that cannot be written only with the sum of toric inequalities.

We apply the proof methods discussed in section \ref{sec:proof-methods}. However, we need a hybrid contraction map for the latter example that geometrically determines the geometric bitstrings $\tilde{F}(X')$ and numerically determines the non-geometric bitstrings $\tilde{F}_{ng}$\footnote{The completion of such partially completed contraction maps may be efficiently achieved by the rules described in \cite{Bao:2024contraction_map}.}.

%Redundant means one can reduce an inequality into a sum of toric inequalities.

%Non-redundant means it cannot be decomposed.

\subsection{Redundant inequalities implied only by the toric inequalities}
We study the following subclass of conjectures and prove that they are redundant HEIs, which can be expressed as a sum of toric inequalities,
\begin{equation}\label{eq:inequality-gen-toric}
\begin{split}
    \sum_{i_1,\cdots, i_{n_\alpha}=1}^{\alpha,\cdots,\alpha} \sum_{j_1,\cdots, j_{n_\beta}=1}^{\beta,\cdots,\beta} S_{a^+_{i_1}\cdots a^+_{i_{n_\alpha}} b^-_{j_1}\cdots b^-_{j_{n_\beta}}}  \geq \sum_{i_1,\cdots, i_{n_\alpha}=1}^{\alpha,\cdots,\alpha} &\sum_{j_1,\cdots, j_{n_\beta}=1}^{\beta,\cdots,\beta}  S_{a^-_{i_1}\cdots a^-_{i_{n_\alpha}} b^-_{j_1}\cdots b^-_{j_{n_\beta}}}\\
    & + \alpha^{{n_\alpha}-1} \beta^{{n_\beta}-1} S_{A_1 \cdots A_{n_\alpha}}.
    %\sum_i \chi_{A_i} S_{A^{(\alpha_i)}_i} + \sum_j \chi_{B_j} S_{B^{(\alpha_j)}_j} \\
    %& +\cdots
    %f(\{\alpha_x\},\{\beta_y\}; O \in \{A_{(i_s,s)}\})S_{A^{(\alpha_1)}\cdots A^{(\alpha_a)}}\\
\end{split}
\end{equation}
%where the purifier $O$ is in one of the $A$-type regions.
%where $a$ and $b$ are the number of $\alpha$'s and $\beta$'s. 
We have $l=\alpha^{n_\alpha}\beta^{n_\beta}$ terms and $r=l+\alpha^{{n_\alpha}-1}\beta^{{n_\beta}-1}$ terms on the LHS and RHS.

\begin{corollary}
    $(\alpha,\cdots,\alpha;\beta,\cdots, \beta)$-inequalities with the non-geometric terms in (\ref{eq:inequality-gen-toric}) are redundant HEIs implied by $(\alpha;\beta)$-inequalities.
    
\end{corollary}

\begin{proof}

    Using lemma \ref{lem:decomposition-with-cyclegraphs}, $G^{\pm}_{\{\alpha\}} $ decomposes into $\alpha^{n_\alpha-1}$ cycle graphs with the length $lcm(\alpha,\cdots,\alpha) = \alpha$. Similarly, $G^{\pm}_{\{\beta\}} $ decomposes into $\beta^{n_\beta-1}$ cycle graphs with the length $lcm(\beta,\cdots,\beta) = \beta$. Then, following construction \ref{thm:graph-general}, we get $\alpha^{n_\alpha-1}\beta^{n_\beta-1}$ numbers of $2$-torii, i.e.,
    \begin{equation}
        T_\mR = T_{\mR_{(1,1)}} \cup \cdots \cup T_{\mR_{(\alpha^{n_\alpha-1},\beta^{n_\beta-1})}}.
    \end{equation}
    Moreover, there are $\alpha^{n_\alpha-1}\beta^{n_\beta-1}$ non-geometric terms. Thus (\ref{eq:inequality-gen-toric}) can be decomposed into the sum of $\alpha^{n_\alpha-1}\beta^{n_\beta-1}$ numbers of $(\alpha;\beta)$-inequalities.

    Because $(\alpha;\beta)$-inequalities are valid HEIs from corollary \ref{cor:toric-HEIs}\cite{Czech:2023xed}, (\ref{eq:inequality-gen-toric}) are redundant holographic entropy inequalities.

\end{proof}

We give an example of (\ref{eq:inequality-gen-toric}), the $(\alpha_1=3,\alpha_2=3;\beta=3)$-conjecture whose non-geometric terms are fixed by the superbalance conditions.
\begin{equation}
    \sum_{i_1,i_2=1}^{3,3} \sum_{j=1}^{3} S_{a_{i_1}^+b_{i_2}^+ c^-_{j}} \geq \sum_{i_1,i_2=1}^{3,3} \sum_{j_1=1}^{3} S_{a_{i_1}^-b_{i_2}^-c^-_{j}} + 3 S_{a_{1}a_{2}a_{3}b_{1}b_{2}b_{3}},
\end{equation}
which has the following decomposition\footnote{(Color online) Each color corresponds to $(3;3)$-inequality.},
\begin{equation}
\begin{matrix}
    \mathcolor{red}{\sum_{i=1}^{3}\sum_{j=1}^{3} S_{a_{i}^+b_{i}^+c^-_{j} }}\\
    +\mathcolor{blue}{\sum_{i=1}^{3}\sum_{j=1}^{3} S_{a_{i}^+b_{i+1}^+c^-_{j} }}\\
    +\mathcolor{teal}{\sum_{i=1}^{3}\sum_{j=1}^{3} S_{a_{i}^+b_{i+2}^+c^-_{j} }}\\
\end{matrix}
\geq
\begin{matrix}
    \mathcolor{red}{\sum_{i=1}^{3}\sum_{j=1}^{3}  S_{a_{i}^-b_{i}^-c^-_{j}}+S_{a_{1}a_{2}a_{3}b_{1}b_{2}b_{3}}}\\
    +\mathcolor{blue}{\sum_{i=1}^{3} \sum_{j=1}^{3}S_{a_{i}^-b_{i+1}^-c^-_{j}} +S_{a_{1}a_{2}a_{3}b_{1}b_{2}b_{3}} }\\
    +\mathcolor{teal}{\sum_{i=1}^{3}\sum_{j=1}^{3} S_{a_{i}^-b_{i+2}^-c^-_{j}} +S_{a_{1}a_{2}a_{3}b_{1}b_{2}b_{3}}} \\
\end{matrix}.
\end{equation}

\subsection{Other inequalities}

%We present an example of non-redundant inequalities and its proof. In particular, we consider $(5,3;1)$

%In the toric case, the global terms did not depend on the parameters. 

%Here, we study the generalized toric conjectures that cannot be written only with the sum of toric inequalities. We call them the non-redundant conjectures\footnote{One should not be confused with ``primitive'' inequalities.}. %We observe that there are no tighter non-redundant inequalities with a disjoint union of torii than the redundant ones, which can be proved using a geometric contraction map.

%We briefly present an example of non-redundant conjectures and review the proof method discussed in section \ref{sec:generalized-toric}. %We close this section by commenting on the analytically feasible geometric contraction maps.

%To prove the conjectures are valid HEIs, we apply the proof methods discussed in section \ref{sec:generalized-toric}. 
%Using the conditional entropy expression, we have

For the generalized toric conjectures,
\begin{equation}
\begin{split}
    \sum_{i_1,\cdots, i_{n_\alpha}=1}^{\alpha_1,\cdots,\alpha_{n_\alpha}} \sum_{j_1,\cdots, j_{n_\beta}=1}^{\beta_1,\cdots,\beta_{n_\beta}} S_{a^+_{i_1}\cdots a^+_{i_{n_\alpha}} b^-_{j_1}\cdots b^-_{j_{n_\beta}}}  \geq \sum_{i_1,\cdots, i_a=1}^{\alpha_1,\cdots,\alpha_{n_\alpha}} &\sum_{j_1,\cdots, j_b=1}^{\beta_1,\cdots,\beta_{n_\beta}}  S_{a^-_{i_1}\cdots a^-_{i_{n_\alpha}} b^-_{j_1}\cdots b^-_{j_{n_\beta}}}\\
    & + \textit{``non-geometric terms''}
\end{split}
\end{equation}
%Conditional entropies are the geometric terms represented by the toroidal graph constructed above. 
the geometric bitstrings $\tilde{F}(X) = \tilde{f}(x_{(1,1)})\times \cdots \times \tilde{f}(x_{(m,n)})$ are determined geometrically. The non-geometric bitstrings $\tilde{F}_{ng}=\{0,1\}^{r-l}$ corresponding to the ``non-geometric terms" could be determined by the second subrule in Rule \ref{rul:general-toric}. 

In the case of the toric inequalities, or $(\alpha;\beta)$-inequalities, the second subrule in Rule \ref{rul:toric}, i.e., $\|x\|_1=\|f(x)\|_1\text{ mod }2$, determines $\tilde{f}_{ng}$. %simply because Rule \ref{rul:toric} fixes the value of the single bit. 
%In addition, especially for $(\alpha;\beta)$-inequalities, we have
%\begin{equation}
 %   \|x^{a_i}\|_1=\frac{\alpha+1}{2},\;\|\tilde{f}(x^{a_i})\|_1=\frac{\alpha-1}{2}.
%\end{equation}
%Thus,
%\begin{equation}
%   \{\|x^{a_i}\|_1-\|\tilde{f}(x^{a_i})\|_1\} \text{ mod }2 =1
%\end{equation}
%for any $a_i \in A$. Similarly, for any $b_j\in B$,
%\begin{equation}
%   \{ \|x^{b_j}\|_1-\|\tilde{f}(x^{b_j})\|_1 \}\text{ mod }2 =0.
%\end{equation}
%As a result, their boundary conditions, $f(x^{a_i})=y^{a_i}$ and $f(x^{b_j})=y^{b_j}$ for any $i$ and $j$, are satisfied  regardless of subregions. 
On the contrary, some instances of the generalized toric conjectures can have multiple non-geometric terms whose corresponding bitstrings are not necessarily determined solely by Rule \ref{rul:general-toric}. This is the fundamental obstruction resulting in the necessity of hybrid contraction maps.

For example, consider $(3,1;3)$-conjecture for the sets of subregions $A=\{a_1,a_2,a_3\}$, $B=\{b_1\}$, and $C=\{c_1,c_2,c_3\}$. $T_\mR$ of the conjecture has a single torus since $3/lcm(3,1)=1$. %, see figure \ref{fig:}.
%From corollary \ref{cor:decomposition}, we decompose the conjectures. As a result, the bitsrings of the LHS decompose as in (\ref{eq:Xdecomposition}), or,
%\begin{equation}
 %   X=x_{1}\times \cdots \times x_{mn}.
%\end{equation}
Among several variations of $(3,1;3)$-conjectures with different non-geometric terms, we consider the balanced case (\ref{eq:balancepurifierinA}) with $\chi^C=1$ with $b_1$ being a purifier, i.e.,
\begin{equation}\label{eq:313}
\begin{split}
    \sum_{i_1=1}^{3}\sum_{i_2=1}^{1}\sum_{j=1}^{3} S_{a^+_{i_1}b^+_{i_2}c^-_j}   \geq \sum_{i_1=1}^{3}\sum_{i_2=1}^{1}\sum_{j=1}^{3} S_{a^-_{i_1}b^-_{i_2}c^-_j} -``2S_{a_1a_2a_3}" +`` S_{c_1c_2c_3}"+ ``2S_{b_1}"
\end{split}
\end{equation}
where we put double quotations to the non-geometric terms. This is a balanced holographic entropy inequality. Here, we note that (\ref{eq:313}) is the sum of superbalance $(3,1;3)$-conjecture and the subadditivity\footnote{We thank Bart\l{}omiej Czech for the comment on the appearance of subadditivity in (\ref{eq:313sum}).} between $A=\{a_1a_2a_3\}$ and $C=\{c_1c_2c_3\}$, i.e.,
\begin{equation} \label{eq:313sum}
\begin{split}
    \mathcolor{red}{\sum_{i_1=1}^{3}\sum_{i_2=1}^{1}\sum_{j=1}^{3} S_{a^+_{i_1}b^+_{i_2}c^-_j}   \geq \sum_{i_1=1}^{3}\sum_{i_2=1}^{1}\sum_{j=1}^{3} S_{a^-_{i_1}b^-_{i_2}c^-_j} &-``S_{a_1a_2a_3}" +`` 2S_{c_1c_2c_3}" + ``S_{a_1a_2a_3c_1c_2c_3}"}\\
    &-``S_{a_1a_2a_3}"  - ``S_{c_1c_2c_3}" + ``S_{a_1a_2a_3c_1c_2c_3}" 
\end{split}
\end{equation}
where we used $S_{b_1}= S_{a_1a_2a_3c_1c_2c_3}$. The first line is the superbalanced $(3,1;3)$-conjecture (\ref{eq:2p1msuperbalance}) which by itself is an invalid inequality\footnote{(Color online) Superbalance $(3,1;3)$-conjecture is written in red letters.}. The second line is the subadditivity between $A$ and $C$, the addition of which turns the inequality into a valid one.

For the proof, we move the part of non-geometric terms $``2S_{a_1a_2a_3}"$ on the RHS of (\ref{eq:313}) to the LHS because it has a negative coefficient when it is on the RHS. That is,
\begin{equation}
\begin{split}
    \sum_{i_1=1}^{3}\sum_{i_2=1}^{1}\sum_{j=1}^{3} S_{a^+_{i_1}b^+_{i_2}c^-_j} +``2S_{a_1a_2a_3}"    \geq \sum_{i_1=1}^{3}\sum_{i_2=1}^{1}\sum_{j=1}^{3} S_{a^-_{i_1}b^-_{i_2}c^-_j} +`` S_{c_1c_2c_3}"+ ``2S_{a_1a_2a_3c_1c_2c_3}".
\end{split}
\end{equation}

Since we have the part of the non-geometric terms on the LHS, we define
\begin{equation}
    X':=X\times \tilde{X}_{ng}, \; F(X') := \tilde{F}(X')\times \tilde{F}'_{ng}.
\end{equation}
$\tilde{X}_{ng}$ adds an extra bit to the bitstrings of the LHS. The geometric bitstrings are still determined geometrically from Rule \ref{rul:general-toric}.

The non-geometric bitstrings cannot be determined solely by the second subrule,
\begin{equation}
    \|X'\|_1 = \|F(X')\|_1\text{ mod }2,
\end{equation}
because there are choices of bits to assign. For example, 
$\tilde{F}'_{ng}=\{0,1\}^3$ has choices from $\{0,0,1\},\{0,1,0\},\{1,0,0\}$, or $\{1,1,1\}$ when 
\begin{equation}
    \{\|X'\|_1-\|\tilde{F}(X')\|_1\} \text{ mod }2=1.
\end{equation} Similarly, $\tilde{F}'_{ng}=\{0,1\}^3$ has choices from $\{0,0,0\},\{0,1,1\},\{1,0,1\}$, or $\{1,1,0\}$ when
\begin{equation}
    \{\|X'\|_1-\|\tilde{F}(X')\|_1\} \text{ mod }2=0.
\end{equation} Despite the large redundancy, not all choices are independent. We found a contraction map for this inequality using \cite{Bao:2024contraction_map}.

\section{Discussions}\label{sec:discussions}

We conclude by pointing out the possible future approaches and directions.\\

\noindent
\textbf{Interactions among torii.}

We studied the case when $T_\mR$ decomposes\footnote{For a given HEI, in general, the decomposition of a contraction map implies that the inequality is the sum of HEIs. The completeness of a contraction map\cite{Bao:2024contraction_map} states that there exists a contraction map if and only if a HEI exists. Suppose a contraction map of a given HEI decomposes into two contraction maps. Then, the inequality should contain two HEIs by the completeness. These two HEIs can be extracted from the original HEI based on the boundary conditions.} into disjoint torii $T_{\mR_\tau}$ for $\tau=(1,1), \cdots ,(m,n)$. Moreover, we restricted ourselves to the cases when the geometric contraction maps $F$ decompose. 
We say that there are \textit{interactions} among $2$-torii when the geometric contraction maps $\tilde{F}(X)$ of $F(X) = \tilde{F}(X)\times \tilde{F}_{ng}$ do not decompose, i.e.,
\begin{equation}
    \tilde{F}(X = x_{(1,1)}\times \cdots \times x_{(m,n)}) \neq \tilde{f}(x_{(1,1)})\times \cdots \times \tilde{f}(x_{(m,n)})
\end{equation}
This happens when a single bit-flip on a face of $T_{\mR_\tau}$ %, e.g., flipping the color ``0'' of $L_u \in \mL_{\tau} $ 
induces multiple bit-flips on vertices of $T_{\mR_\tau}$ and at least another torus $T_{\mR_{\tau'}}$.

%$\gamma(\alpha;\beta)$ $(\alpha;\beta)$-inequalities, or flipping volume of the atomic cells of $\gamma(\alpha_1,\alpha_2;\beta)$, results in a bit operation on the vertex/vertices of another torus or other torii. Then, we claim the following proposition.

When a single bit-flip results in double bit-flips, a candidate geometric map cannot be a contraction map because $\|X-X'\|_1=1$ implies $\|F(X)-F(X')\|_1 = 2$. One way to fix the issue is to modify the coefficients of the LHS of the conjectures, for instance, multiplying $2$ on all the terms of the LHS, i.e.,
\begin{equation}
\begin{split}
    2 \sum_{i_1,\cdots, i_{n_\alpha}=1}^{\alpha_1,\cdots,\alpha_{n_\alpha}} \sum_{j_1,\cdots, j_{n_\beta}=1}^{\beta_1,\cdots,\beta_{n_\beta}} S_{a^+_{i_1}\cdots a^+_{i_{n_\alpha}} b^-_{j_1}\cdots b^-_{j_{n_\beta}}}  \geq \sum_{i_1,\cdots, i_{n_\alpha}=1}^{\alpha_1,\cdots,\alpha_{n_\alpha}} &\sum_{j_1,\cdots, j_{n_\beta}=1}^{\beta_1,\cdots,\beta_{n_\beta}}  S_{a^-_{i_1}\cdots a^-_{i_{n_\alpha}} b^-_{j_1}\cdots b^-_{j_{n_\beta}}}\\
    & + \textit{``non-geometric terms''}
\end{split}
\end{equation}
Then, $\|X-X'\|_1=2$ implies $\|F(X)-F(X)'\|_1=2$ because a single bit-flip on a face of $T_\mR$ now has a weight $2$. We hope to explore further variations of these conjectures.

\noindent\\
\textbf{A geometric contraction map from connected component of a graph.}

We discuss another proof method to explore the generalized toric conjectures without considering the decompositions. First, we reinterpret the contractible knots as geometrically closed partitions of vertices of $T_\mR$. We define a graph $G_{path}:=(V_\mR,E_{path})$ where $V_\mR$ is the set of vertices of $T_\mR$. $E_{path}$ is the set of edges that connect vertices unless the edges cross the partition. This defines the disjoint subsets of vertices of $G_{path}$. Each disjoint subset and the subgraph induced by it can be referred to as a path-connected subset and connected component, respectively. Then, $\|f(x)-f(x')\|_1$ measures the change in the number of path-connected subsets in $G_{path}$.

%to generalize them in higher dimensional cases. Rule \ref{rul:toric} \cite{Czech:2023xed} gave rules to assign ``1'' to the vertices of $T_\mR$ when they are enclosed by the contractible knots. Thus, the contractible knots partition the set of vertices of $T_\mR$.% This is when the vertices of $T_\mR$ are partitioned, % from the dual vertices
%see figure. 
%Then, 

\begin{figure}
    \centering
    \includegraphics[width=0.8\linewidth]{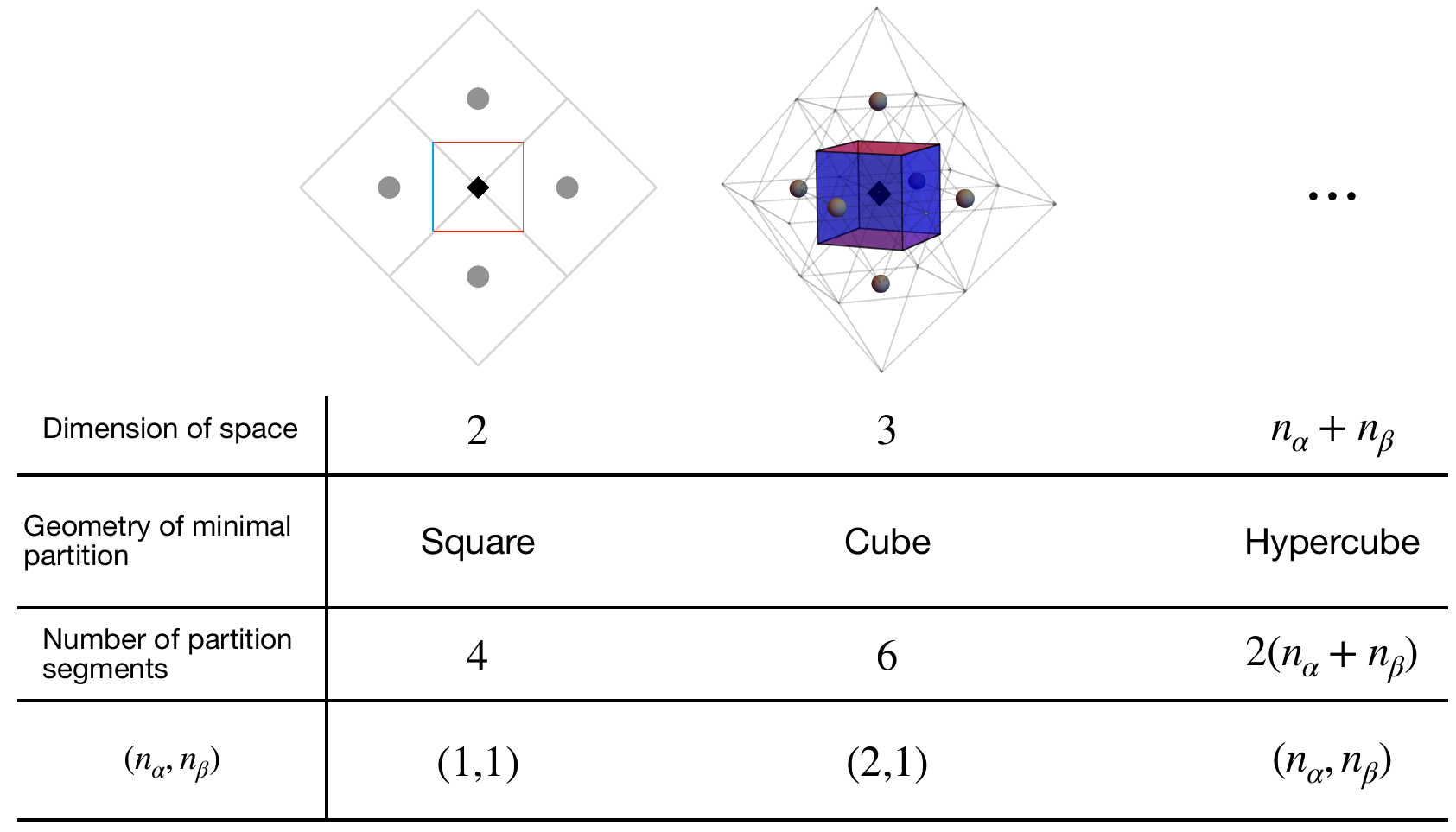}
    \caption{\small{Minimal partitions of graph $T_\mR$ embedded in $(n_\alpha+n_\beta)$-dimensional space for the generalized toric conjectures consists of $2(n_\alpha+n_\beta)$ hypersurfaces. The vertex of $T_\mR$ surrounded by the minimal partition is marked with the black rhombus. The dual vertices of $T_\mL$ are marked with solid spheres. When $(n_\alpha,n_\beta)=(1,1)$ and $n_\alpha+n_\beta=2$, the minimal partition is a square and consists of $2n_\alpha=2$ blue vertical line segments and $2n_\beta=2$ red horizontal line segments. The left and right rhombi are colored with $1$ as in figure \ref{fig:bitstringgeometry}. The top and bottom rhombi are colored with $0$. When $(n_\alpha,n_\beta)=(2,1)$ and $n_\alpha+n_\beta=3$, the minimal partition is a cube and consists of $2n_\alpha=4$ blue vertical squares and $2n_\beta=2$ red horizontal squares. }}
    \label{fig:minimalpartition}
\end{figure}

To construct partitions for the generalized toric conjectures, it is useful to notice that the partition segments are placed between the vertex in $T_\mR$ and its dual vertex in $T_\mL$. Then, the number of segments should match the number of dual vertices around a vertex, see figure \ref{fig:minimalpartition}.

For simplicity, we consider a minimal partition as an example here without loss of generality. A minimal partition encloses only a single vertex of $T_\mR$. For example, a minimal partition comprises four line segments in $(\alpha;\beta)$-inequalities. In the case of the generalized toric conjectures without the decompositions, $T_\mR$ is a $(n_\alpha+n_\beta)$-dimensional toroidal graph. We can see that its minimal partition is a $(n_\alpha+n_\beta)$-dimensional hypercube because the number of codimension-$1$ surfaces of a $(n_\alpha+n_\beta)$-dimensional hypercube matches with the number of dual vertices of $T_\mL$. The number of codimension-$1$ surfaces of a $(n_\alpha+n_\beta)$-dimensional hypercube is $2(n_\alpha+n_\beta)$\cite{coxeter1973regularpolytopes}. The number of dual vertices of vertex $R_v$ of $T_\mR$ is given by, from lemma \ref{lem:cardinalities-general}, 
\begin{equation}
    |\overline{Inc}(R_v)|+|Exc(R_v)|   = 2({n_\alpha}+{n_\beta}) - \sum_{s,t=1}^{{n_\alpha},{n_\beta}} (\delta_{\alpha_s,1} +\delta_{\beta_t,1}).
\end{equation}
The Kronecker deltas in the last term indicate the identification of codimension-$1$ surfaces. Hence, the number of dual vertices matches the number of codimension-$1$ surfaces.

In general, the partitions, not necessarily the minimal partitions, determine $G_{path}$ of $T_\mR$. Hence, $\|F(X)-F(X')\|_1$ can be understood as the change in the number of the path-connected subsets in $G_{path}$. The proof method in \cite{Czech:2023xed} and this paper can potentially be reformulated using the following geometric assignment of bitstrings of the LHS of the generalized toric conjectures.

%Now, the contractible knots can be understood as geometrically closed partitions. The proof methods in \cite{Czech:2023xed}, and this paper can potentially be reformulated using the following geometric assignment of bitstrings of the LHS.

%This motivates us to define new geometric assignment as follows. 

%One can define the connected components of $T_\mR$ by defining a subgraph  Therefore, $\|F(X)-F(X')\|_1=1$ could be understood as the change in the number of connected components.

\noindent\\
\underline{\textit{Geometric assignment of bitstrings $X$}}\\

Consider $(\alpha_1,\cdots,\alpha_{n_\alpha};\beta_1,\cdots,\beta_{n_\beta})$-conjectures. 
\begin{itemize}
    \item For each element of bitstrings $X$ of the LHS, assign codimension-$1$ surfaces of a $(n_\alpha+n_\beta)$-dimensional hypercube. In particular, on each $L_u$, assign
    \begin{equation}
        \begin{cases}
        \text{$2n_\alpha$ hypersurfaces} & \text{if }X_u =1\\
        \text{$2n_\beta$ hypersurfaces} & \text{if }X_u =0\\
        \end{cases}
    \end{equation}
    
    %\item Define a subgraphs $G_c$ of $T_\mR$, i.e.,
    %\begin{equation}
    %    V_D, E_D.
    %\end{equation}
    %In other words, a subset of vertices are adjacent only when the edges connecting the vertices do not cross the knots. Hence, there is a set of disconnected subgraphs. Knots define the cuts of connected components.
    %\item $\|x-x'\|_1=1$ corresponds to a single bit-flip.
    %\item $\|f(x)-f(x')\|_1=1$ implies the change in the number of disconnected subgraphs. 
    
\end{itemize}
We hope to explore the proof methods of the generalized toric conjectures that are fully geometric instead of the hybrid method discussed in section \ref{sec:examples}.

\noindent\\
\textbf{Generalization of $\mathbb{RP}^2$ inequalities.}

%generalized RP2 inequalities, RP2 and toric share MMI
%relation to generalized RP2 inequalities

In addition to toric inequalities, \cite{Czech:2023xed} also introduced the so-called $\mathbb{RP}^2$-inequalities, which have a graphical representation on a projective plane. It is our future interest to explore generalizations of $\mathbb{RP}^2$ inequalities and the properties of their graphs, such as the interplay between generalized real projective plane conjectures and generalized toric inequalities and conjectures.

\noindent\\
\textbf{Geometries of holographic entropy inequalities.}

Interestingly, \cite{Czech:2023xed} found a way to geometrize HEIs, such as toric inequalities as a torus, $\mathbb{RP}^2$-inequalities as a projective $2$-plane and $i_6$\cite{Sergio2019:5regions} as a hyperbolic disk. In our program to fully understand the holographic entropy cones for higher parties, another future direction is to construct and classify the geometries of other known thousands of HEIs \cite{Bao:2015bfa,Sergio2019:5regions,Czech:22seven-party,HEarrangement2019Hubeyetal,Hernandez-Cuenca:2023iqh}. However, we currently do not fully understand the criteria for HEIs to have geometric counterparts.

%Moreover, the geometric assignment 

%Geometries of holographic entropy inequalities are based on the EWN relation. This can be defined physical bits and unphysical bits, 

%upper set 

%hysical meaning of genus, higher genus inequlaity, necessary condition that entropy inequalities has geometric representations. relations to hec.

%\noindent\\
%\textbf{Physical meaning of a contraction map.}

%\noindent\\
%\textbf{Symbolic dynamics and entropy cone program.}

%\noindent\\
%\textbf{Quantum error correction codes on a higher dimensional toroidal graph.}

%Toric code for the ground state from our constructions. %Particle excitations, disconnecting the connected components excites two particles. Connecting the the disconnected components de-excites the particles.

%Higher dimensional toric codes

%graph complement, virtual spins on vertices of $T_\mR$, toric codes

%The relation between the graphs of the toric inequalities and the toric code\cite{KITAEV20062} was briefly discussed in \cite{Czech:2023xed}. Our construction of the graphs of the inequalities and conjectures could provide the lattice 

\section*{Acknowledgement}
We thank Bart\l{}omiej Czech, Sirui Shuai, Yixu Wang, Dachen Zhang, Dimitrios Patramanis, Nima Lashkari, Shoy Ouseph, Mudassir Moosa, Kwing Lam Leung, Fabian Ruelle and Paul Oehlmann for the discussions. We extend gratitude to Bart\l{}omiej Czech for comments on the manuscript. K.F. is grateful for the invitation to visit Bart\l{}omiej Czech's group and the hospitalities at Tsinghua University. N.B. is funded by the Quantum Telescope Project. K.F. is supported by N.B.'s startup funding at Northeastern University. J.N. is partially supported by the NSF under Cooperative Agreement PHY2019786 and N.B.'s startup funding at Northeastern University.

%any grant??

%\appendix

%\section{Global constraints and boundary conditions}

%Here, we explore the possible choices of determining the last pieces, which are consistent to the boundary conditions. 

%\section{Detailed and numerical proof of $(3,3;1)$ and its extension to $(3,3;\beta)$}

\bibliographystyle{JHEP}
\bibliography{main}

\end{document}